\documentclass[12pt]{article}

\usepackage[english]{babel}
\usepackage{graphicx}
\usepackage{framed}
\usepackage[normalem]{ulem}
\usepackage{amsmath}
\usepackage{amsthm}
\usepackage{amssymb}
\usepackage{amsfonts}
\usepackage{colonequals}
\usepackage{enumerate}
\usepackage{extarrows}
\usepackage{tikz-cd}
\usepackage[utf8]{inputenc}
\usepackage{indentfirst}
\usepackage{dirtytalk}
\usepackage{csquotes}
\usepackage{hyperref}
\usepackage{nicefrac, xfrac}
\usepackage{dutchcal} 
\usepackage[top=1 in,bottom=1in, left=0.8 in, right=0.8 in]{geometry}

\theoremstyle{plain}
\newtheorem{thm}{Theorem}[section]
\newtheorem{lemma}[thm]{Lemma}
\newtheorem{cor}[thm]{Corollary}
\newtheorem{prop}[thm]{Proposition}

\theoremstyle{definition}
\newtheorem{defn}[thm]{Definition}
\newtheorem{example}[thm]{Example}

\newtheorem{remark}[thm]{Remark}

\newcommand\C{\mathcal{C}}

\newcommand\NN{\mathbb{N}}

\newcommand\U{\mathcal{U}}
\newcommand{\Ty}{\mathtt{Ty}}
\newcommand{\Tm}{\mathtt{Tm}}
\newcommand{\pair}{\operatorname{\mathsf{pair}}}
\newcommand{\inl}{\operatorname{\mathsf{inl}}}
\newcommand{\inr}{\operatorname{\mathsf{inr}}}
\newcommand{\zero}{\operatorname{\mathsf{0}}}
\newcommand{\UA}{\operatorname{\mathsf{UA}}}
\newcommand{\funext}{\operatorname{\mathsf{funext}}}
\newcommand{\UIP}{\operatorname{\mathsf{UIP}}}
\newcommand{\suc}{\operatorname{\mathsf{succ}}}
\newcommand{\refl}{\operatorname{\mathsf{refl}}}
\newcommand{\app}{\operatorname{\mathsf{app}}}
\newcommand{\isFibrant}{\operatorname{\mathsf{isFibrant}}}
\newcommand{\id}{\operatorname{\mathsf{id}}}
\newcommand{\tr}{\operatorname{\mathsf{tr}}}
\newcommand{\ap}{\operatorname{\mathsf{ap}}}
\newcommand{\happly}{\operatorname{\mathsf{happly}}}
\newcommand{\unit}{\operatorname{\mathbf{1}}}
\newcommand{\emptype}{\operatorname{\mathbf{0}}}
\newcommand{\exotoid}{\operatorname{\mathsf{eqtoid}}}
\newcommand{\Id}{\operatorname{\mathsf{Id}}}
\newcommand{\Eq}{\operatorname{\mathsf{Eq}}}
\newcommand{\exoeq}{\operatorname{\mathsf{=}^\textit{e}}}
\newcommand{\ecirc}{\operatorname{\mathsf{\circ}^\textit{e}}}
\newcommand{\List}{\operatorname{\mathsf{List}}}
\newcommand{\BinTree}{\operatorname{\mathsf{BinTree}}}
\newcommand{\Parens}{\operatorname{\mathsf{Parens}}}
\newcommand{\UnLBinTree}{\operatorname{\mathsf{UnLBinTree}}}
\newcommand{\nil}{\operatorname{\mathsf{[]}}}
\newcommand{\cons}{\operatorname{\mathsf{\coloncolon}}}
\newcommand{\leaf}{\operatorname{\mathsf{leaf}}}
\newcommand{\uleaf}{\operatorname{\mathsf{u-leaf}}}
\newcommand{\unode}{\operatorname{\mathsf{u-node}}}
\newcommand{\node}{\operatorname{\mathsf{node}}}
\newcommand{\popen}{\operatorname{\mathsf{popen}}}
\newcommand{\pclose}{\operatorname{\mathsf{pclose}}}
\newcommand{\isbalanced}{\operatorname{\mathsf{isbalanced}}}
\newcommand{\iscontr}{\operatorname{\mathsf{isContr}}}
\newcommand{\Balanced}{\operatorname{\mathsf{Balanced}}}
\newcommand{\coer}{\mathcal{c}}
\newcommand{\Set}{\operatorname{\mathsf{Set}}}
\newcommand{\SSet}{\operatorname{\mathsf{SSet}}}
\newcommand{\op}{\operatorname{\text{op}}}

\title{Formalizing two-level type theory \\ with cofibrant exo-nat}
\author{Elif Uskuplu}
\date{}

\begin{document}

\maketitle

\begin{abstract}
   This study provides some results about two-level type-theoretic notions in a way that the proofs are fully formalizable in a proof assistant implementing two-level type theory such as Agda. The difference from prior works is that these proofs do not assume any abuse of notation, providing us with more direct formalization. Moreover, some new notions, such as function extensionality for cofibrant types, are introduced. The necessity of such notions arises during the task of formalization. In addition, we provide some novel results about inductive types using cofibrant exo-nat, the natural number type at the non-fibrant level. While emphasizing the necessity of this axiom by citing new applications as justifications, we also touch upon the semantic aspect of the theory by presenting various models that satisfy this axiom.
   
\end{abstract}

\tableofcontents

\section{Introduction}

Formalizing mathematics has a newer and more vivid meaning than it had before. It stands for the task of representing mathematical knowledge in digital form. \emph{Formalization} in the new sense is close to \emph{implementation of a theory} or \emph{verification of a knowledge}. The basic tools for such a task are the proof assistant programs like Lean, Coq, Agda, and more \cite{PAlist}. In order to formalize a piece of mathematics, we should provide equational reasoning and basic definitions of the mentioned theory to our favorite proof assistant, and then we can make formal derivations relying on the previous building blocks. These derivations might be the digitized version of the knowledge that we proved before by hand. In this sense, the task is a kind of verification of the human knowledge done on paper. Moreover, we can see the implementation aspect of the formalization in the proof assistant itself. For example, proof assistants such as Lean, Coq, and Agda implement dependent type theory \cite{proofas} while Mizar implements Tarski–Grothendieck set theory \cite{mizar}.

As proof assistants become more interactive than before, it is possible to learn more from the formalization process itself. For many mathematicians, using \emph{abuse of notation} is a natural approach, and it indeed has some benefits. However, when we start to make mathematics precise in a proof assistant, this approach is not allowed. This is one of the significant differences between mathematics on paper and on computer. Therefore, during the formalization task, it is very likely to explore the gaps due to abuse of formal language. This is an excellent example of what we can learn from proof assistants. From the author's perspective, it is fair to say that as we teach computers to be clever at mathematics, they teach us to be \emph{more clever}. This short study is an experience of the interaction between a theory discovered by people and its formalization. 

The theory subject to the experience is the Homotopy type theory (HoTT). It is a new foundational theory for mathematics. It relies on the intuitionistic type theory with a homotopical interpretation. This theory is also known as the Univalent Foundation due to its essential part, the Univalence axiom. It roughly says equivalent mathematical objects are equal. However, it lacks modeling of some structures (e.g., semisimplical types). Thus, there are some efforts to extend HoTT. Two-level type theory (2LTT) is one of these extensions\footnote{While we provide basics of 2LTT, we mainly refer to \cite{hott} for HoTT and to \cite{2ltt} for 2LTT.}. Briefly saying, 2LTT has two levels; the base level is HoTT while the second level is a traditional form of type theory validating the uniqueness of identity proofs. One can think that the second level is the meta-theory of the first. Section \ref{synchap} and Section \ref{novel-sec} of this paper aim to analyze the bridge between the two levels in terms of formalized mathematics. 

Why do we care about this bridge, although it has already been analyzed on paper? As a reader who is now familiar with proof assistants might guess, the formalization task is not easy compared to the work done on paper. Some definitions need to be changed or adjusted to be applicable to the assistant. Even obvious derivations should be implemented to obtain precise proofs. In other words, there is no room for gaps. In the case of 2LTT, there is no well-accepted formalization for now. Recently, one of the proof assistants, Agda, has released some new features that allow us to work with 2LTT. Moreover, using these, we have developed an Agda library \cite{agdalib} about 2LTT and some of its applications. This was one of the first attempts to use these features of Agda. Although the initial goal was to formalize the content of the paper \emph{The Univalence Principle} \cite{UPpaper}, the basics of 2LTT had to be built first because the study in the mentioned paper is based on 2LTT. Within this experience, some modifications to the definitions and some additional tools were needed. One of our goals is to emphasize these changes and additions that make 2LTT applicable in Agda easily.

During our Agda project, we encountered situations where certain proofs required the formalization of a new auxiliary tool, which we refer to as \emph{function extensionality for cofibrant types}. Function extensionality is a fundamental property of dependent functions, asserting that two functions are equal if and only if they produce equal results for every input. The specific notion of equality may vary depending on different contexts and levels, but in the case of traditional function extensionality, the equality notion remains consistent both in the domain and the range. However, when dealing with cofibrant types, the situation is different. Here, the equality notions for the input terms and the output terms may differ. Therefore, in our study, we introduce a novel function extensionality property tailored for such cases, and we rigorously establish its validity. Furthermore, our project led us to uncover novel results related to certain inductive types, notably \emph{List} and \emph{Binary-Trees}, which had not been explored within the context of 2LTT before. What initially started as a foundation for another study has opened up exciting new directions for further research.

One of the original motivations for 2LTT was to define semisimplicial types. However, although plain 2LTT allows defining the type of $n$-truncated semisimplicial types for any exo-natural number $n$, a term of $\NN$ in the second level, it does not seem possible to assemble these into a type of untruncated semisimplicial types. Voevodsky’s solution \cite{hts} was to assume that exo-nat, $\NN$ in the second level, is fibrant (isomorphic to a type in the first level), which works for simplicial sets but may not hold in all infinity-toposes. However, assuming cofibrancy, a weaker notion than fibrancy, of exo-nat also allows for defining a fibrant type of untruncated semisimplicial types with a broader syntax, including models for all infinity-toposes. After giving the overview of the models of 2LTT in Section \ref{semchap}, we provide such models in Section \ref{modelsec}.

\textbf{Structure of this work.} In Section \ref{synchap}, we begin with giving the basics of 2LTT. Our basic objects, \emph{types} and \emph{exo-types} are explained. We then give the three classifications about exo-types, which are \emph{fibrancy}, \emph{cofibrancy}, and \emph{sharpness}. Note that these concepts are the basic building blocks of the mentioned study \cite{UPpaper}. We also provide new results about the cofibrancy and sharpness of some inductive types. Proposition \ref{funext-cofib} and the entire Section \ref{novel-sec} are new in this field. Throughout the paper, we point to the relevant codes in the Agda library and talk about how, if any, things that differ from previous works contribute to Agda formalization. In Section \ref{semchap}, in order to present the complete picture, we also explore the semantic aspect of the study and introduce the meaning of 2LTT's model, providing results about the general models of the theory we are concerned with. As far as we know, there have been no previous studies on non-trivial models of 2LTT with cofibrant exo-nat. By \emph{non-trivial}, we mean the proposed model indeed satisfies cofibrant exo-nat but does not satisfy fibrant exo-nat. Theorem \ref{cofib-model} proves the existence of models we desired.

\textbf{Drawback and limitations.} Although the proofs in the paper are logically valid and complete, the formalization of 2LTT heavily depends on new, experimental, and undocumented features of Agda. As such, there are some bugs emerging from the previously untested interactions of these features, and there might be more than we encountered. There are some efforts by Agda developers to fix these bugs in the Agda source code. We expect the study with these experimental features to produce documentation on what we need to avoid bugs. 

\textbf{Acknowledgements}. We would like to thank Michael Shulman and Nicolai Kraus for many interesting discussions and insightful comments. The work is partially supported by NSF grant DMS-1902092, the Army Research Office W911NF-20-1-0075, and the Simons Foundation. The work is also based upon work supported by the Air Force Office of Scientific Research under award number FA9550-21-1-0009.

\section{Review about two-level type theory}\label{synchap}

\subsection{Types \& exo-types}\label{types&exo-types}

The primitive objects of a type theory are \textbf{types} and \textbf{terms}. These are similar to sets and elements in set theory. For 2LTT, there are two different \emph{kinds} of types: one kind in HoTT and other kind in meta level. We reserve the word ``types" for ones in HoTT (as usual) while we use the word ``exo-type\footnote{This term was originally suggested by Ulrik Buchholtz.}" for ones in meta level, as in \cite{UPpaper}. According to this distinction, we should define each type and type formers twice: one for types, one for exo-types. 

In type theory, we define \textbf{universe} as a type of types. In order to avoid paradoxes a la Russell, we assume a universe hierarchy. Thus, a universe is again a type, but in a different sense than its terms.
In our setting, we have a hierarchy of universes of types, denoted by $\U$, and exo-universes of exo-types, denoted by $ \U^e $. We always make the distinction between types and exo-types using the superscript $-^e$.

After having universes and exo-universes, it is easy to define types and exo-types. We are assuming all definitions in HoTT Book \cite{hott}, and hence we have basic type and type formers. Exo-type and exo-type formers are defined exactly in the same way, but these are defined in the exo-universe. 

\begin{defn}\label{basicdef}
\begin{itemize}
\item[]
\item For a type $A : \U$ and a type family $B : A \rightarrow \U$, we define the \textbf{dependent function type} (briefly $\prod$-type) \[\prod_{a : A} B(a)\] as usual. If $B$ is a constant family, then the dependent function type is the ordinary function type: \[\prod_{a : A} B := A \rightarrow B \,.\]

For an exo-type $A : \U^e$ and an exo-type family $B : A \rightarrow \U^e$, we have the \textbf{dependent function exo-type} (briefly $\prod$-exo-type) \[{\prod_{a : A}}^e B(a)\] in a similar way. If $B$ is constant, then we have the ordinary function exo-type \[{\prod_{a : A}}^e B := A \rightarrow^e B \,.\]

It should be noted that the notation for maps between exo-types ``$\rightarrow^e$" can be used throughout this paper to emphasize distinction. However, we omit the notation and use usual arrows for any cases since the domain and the codomain can be derived from the context, or we can specify whether we have type or exo-type.

\item For a type $A : \U$ and a type family $B : A \rightarrow \U$, we define the \textbf{dependent sum type} (briefly $\sum$-type) \[\sum_{a : A} B(a)\] as usual, and its terms are of the form $\pair (a,b)$ for $a : A$ and $b : B(a)$. The projection maps are $\pi_1 : \sum_{a : A} B(a) \rightarrow A $ and $\pi_2 : \sum_{a : A} B(a) \rightarrow B(a)$. When $B$ is a constant family, we call it the \textbf{product type} and denote it by $A \times B$. 

For an exo-type $A : \U^e$ and an exo-type family $B : A \rightarrow \U^e$, we have the \textbf{dependent sum exo-type} (briefly $\sum$-exo-type) \[{\sum_{a : A}}^e B(a)\] in a similar way, and its terms are of the form $\pair^e (a,b)$ for $a : A$ and $b : B(a)$. The projection maps are ${\pi_1}^e : {\sum_{a : A}}^e B(a) \rightarrow A $ and ${\pi_2}^e : {\sum_{a : A}}^e B(a) \rightarrow B(a)$. When $B$ is a constant family, we have the \textbf{product exo-type} $A \times^e B$. 

Note that we will use the notation $\mathbf{(a,b)}$ for $\pair (a,b)$ or $\pair^e(a,b)$ when the context is clear. We prefer the comma notation in this paper due to easier reading. \emph{This choice and the choice for arrows may seem to be contradictory with our claim of ``no abuse of notation". However, the choices are only for aesthetic purposes, and the notation difference is precise in the formalization.}

\item For a pair of types $A , B : \U$, we define the \textbf{coproduct type} $A + B : \U$ as usual, constructed by the maps $\inl : A \rightarrow A + B$ and $\inr : B \rightarrow A + B$. 

For a pair of exo-types $A , B : \U^e$, we define the \textbf{coproduct exo-type} $A +^e B : \U^e$ similarly, constructed by the maps ${\inl}^e : A \rightarrow A +^e B$ and ${\inr}^e : B \rightarrow A +^e B$.

\item While the \textbf{unit type}, denoted by $\unit : \U$, is constructed by a single term $\star : \unit$, the \textbf{unit exo-type}, denoted by $\unit^e : \U^e$, is constructed by a single exo-term $\star^e : \unit^e$. 

\item We have both the \textbf{empty type}, denoted by $\emptype : \U$, and the \textbf{empty exo-type}, denoted by $\emptype^e : \U^e$. Both have no constructors, and hence no term by definition.

\item The \textbf{natural number type}, denoted by $\NN : \U$, is constructed by a term $\zero : \NN$ and a function term $\suc : \NN \rightarrow \NN$. The \textbf{natural number exo-type} (briefly exo-natural or exo-nat), denoted by $\NN^e : \U^e$, is constructed by $\zero^e : \NN^e$ and $\suc^e : \NN^e \rightarrow \NN^e$.

\item The \textbf{finite type} having $n$ terms, denoted by $\NN_{<n}$, defined inductively (on $n:\NN$) as \[\NN_{<0}:= \emptype\quad \text{ and } \quad \NN_{<n+1}:=\NN_{<n}+\unit. \] Similarly \textbf{exo-finite exo-type} having $n$ terms, denoted by $\NN^e_{<n}$, is defined inductively (on $n:\NN^e$) as \[\NN^e_{<0}:= \emptype^e\quad \text{ and } \quad \NN^e_{<n+1}:=\NN^e_{<n}+^e\unit^e. \]

\item For a type $A : \U$ and $a,b : A$, we define the \textbf{identity type} (or \textbf{path type)} $a = b : \U$ as usual, its constructor is $\refl : a = a$. 

For an exo-type $A : \U^e$ and $a,b : A$, we have the \textbf{exo-equality} $a \exoeq b : \U^e$ in a similar way; its constructor is $\refl^e : a \exoeq a$.
\end{itemize}
\end{defn}

Note that these type/exo-type pairs may not coincide in the cases of $+^e$, $\NN^e$, $\emptype^e$, and $\exoeq$. For example, even if $A,B:\U$, we may not have $A +^e B : \U$, namely, this is always an exo-type, but not generally a type. The difference in these cases is that the elimination/induction rules of fibrant types cannot be used unless the target is fibrant. For example, we can define functions into any exo-type by recursion on $\NN^e$, but if we want to define a function $f :\NN \rightarrow A$ by recursion, we must have $A : \U$.

\begin{remark}
We assume the univalence axiom ($\UA$) only for the identity type. Thus, we also have the function extensionality ($\funext$) for it because $\UA$ implies $\funext$ (Theorems 4.9.4 \& 4.9.5 in \cite{hott}). For the exo-equality, we assume the $\funext^e$ and the axiom called Uniqueness of Identity Proofs ($\UIP$). In other words, we have the following 
\begin{itemize}
\item $\UA : \prod_{A,B : \U} (A \simeq B) \rightarrow (A = B)$
\item $\funext : \left(f,g : \prod_A B(a) \right) \rightarrow \left( \prod_{a : A} f(a)=g(a) \rightarrow (f=g) \right)$
\item $\funext^e : \left(f,g : \prod^e_A B(a) \right) \rightarrow \left( \prod^e_{a : A} f(a)\exoeq g(a) \rightarrow (f\exoeq g) \right)$
\item $\UIP : \prod^e_{a,b :A} \left( \prod^e_{(p,q : a\exoeq b)} p \exoeq q \right)$
\end{itemize}
In the applications or examples, we often make use of both versions of function extensionality. Note also that $\UIP$ says that for any terms $a,b$ in an exo-type $A$, if they are exo-equal, namely, there is an exo-equality between them, then the equality term is unique.
\end{remark}

\begin{remark}
Just as there is a type hierarchy in terms of path types such as \textbf{contractible} types, \textbf{propositions}, and \textbf{sets}, we can define \textbf{exo-contractible} exo-type, \textbf{exo-propositions}, and \textbf{exo-sets} similarly with respect to $\exoeq$. Since we assume $\UIP$ for exo-equality, this yields that all exo-types are exo-sets. As another note, any property of $=$ can be defined for $\exoeq$ similarly by its elimination rule. For example, we have both transport ($\tr$) and exo-transport ($\tr^e$), we have both path-type homotopies ($\sim$) of functions between types and exo-equality homotopies ($\sim^e$) of functions between exo-types, and so on. For this kind of properties, not defined here, we refer to the HoTT Book \cite{hott}.
\end{remark}

\textbf{Agda Side.} The folders \texttt{Types} and \texttt{Exo-types} in our Agda library \cite{agdalib} contain all the definitions above. The main file \texttt{Primitive.agda} (Figure \ref{universe-agda}) has the definition of the universe and the exo-universe. The flag \texttt{--two-level} enables a new sort called \texttt{SSet}. This provides two distinct universes for us. Note that while it is common to assume typical ambiguity\footnote{HoTT Book, Section 1.3} in papers, the formalization works with polymorphic universes.
\begin{figure}[ht]
    \centering
        \includegraphics[scale=0.5]{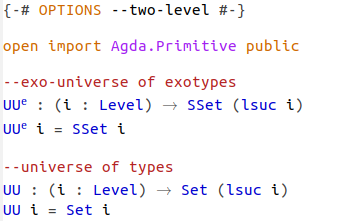}
    \caption{Agda code for two kinds of universes.}
    \label{universe-agda}
\end{figure}

\subsection{Isomorphisms \& Equivalences}

Considering these twin definitions in the previous section, it's natural to ask whether there is a correspondence between them. We obtain such a correspondence according to the relation between types and exo-types. In \cite{2ltt}, it is assumed that there is a coercion map $c$ from types to exo-types, for any type $A : \U$ we have $c(A) : \U^e$. Another approach, as in \cite{UPpaper}, is taking $c$ as an inclusion, in other words, assuming every type is an exo-type. In this work, the second approach is assumed. Therefore, we can apply exo-type formers to types. For example, both $\NN + \NN$ and $\NN +^e \NN$ make sense, but both are still exo-types. We will later prove some isomorphisms related to such correspondences. However, what an isomorphism between exo-types means should be defined beforehand. 

\begin{defn}
\begin{itemize}
\item[]
\item[•] A function $f : A \rightarrow B$ between exo-types is called an \textbf{isomorphism} (or \textbf{exo-isomorphism}) if there is a function $g : B \rightarrow A$ such that $g \ecirc f \exoeq \id_A$ and $f  \ecirc g \exoeq \id_B$ where $\id_A : A \rightarrow A$ is the identity map. We define the exo-type of exo-isomorphisms as \[A \cong B := {\sum}^e_{f:A\rightarrow B}{\sum}^e_{g:B\rightarrow A}(f  \ecirc g \exoeq \id_B)\times^e(g \ecirc f \exoeq \id_A).\] It can be read as $A \cong B$ consists of exo-quadruples $(f,g,p,q)$ such that $f:A\rightarrow B$ and $g:B\rightarrow A$ are functions, and $p, q$ are witnesses for the relevant identities. Note that $\ecirc$ means that the composition is between two functions of exo-types.

\item[•] A function $f : A \rightarrow B$ between types is called an \textbf{equivalence} if its fibers are contractible. We define the type of equivalences as \[A \simeq B := \sum_{f:A\rightarrow B} \left( \prod_{b:B} \texttt{is-Contr}\left(\sum_{a:A} f(a)=b\right)\right).\] It can be read as $A\simeq B$ consists of pairs $(f,p)$ where $f$ is a function and $p$ is a witness of that all fibers of $f$ is contractible. In other words, the preimage of each term in $B$ is unique up to the identity type.

\item[•] A function $f : A \rightarrow B$ between types is called \textbf{quasi-invertible} if there is a function $g : B \rightarrow A$ such that $g \circ f = \id_A$ and $f  \circ g = \id_B$ where $\id_A : A \rightarrow A$ is the identity map. 
\end{itemize}
\end{defn}

\begin{remark}
In these definitions, one can use $\funext^e$ or $\funext$, and instead of showing, for example, $g \ecirc f \exoeq \id_A$, it can be showed that $g (f (a)) = a$ for any $a \in A$. Moreover, a map is an equivalence if and only if it is quasi-invertible. Therefore, we can use both interchangebly. For practical purposes, when we need to show f is an equivalence, we generally do it by showing that it is quasi-invertible.
\end{remark}

Assuming that each type is an exo-type, and considering all definitions so far, the correspondence between exo-type formers and type formers can be characterized\footnote{This is the same as Lemma 2.11 in \cite{2ltt}.} as follows:

\begin{thm}\label{lemma211}
If $A, C: \U$ are types and $B : A \rightarrow \U$ is a type family, we have the following maps. The first three maps are exo-isomorphisms.
\begin{itemize}
\item[i.] $\unit^e \rightarrow \unit$,
\item[ii.] $\sum^e_{a:A} B(a) \rightarrow \sum_{a:A} B(a)$,
\item[iii.] $\prod^e_{a:A} B(a) \rightarrow \prod_{a:A} B(a)$,
\item[iv.] $A +^e C \rightarrow A + C$,
\item[v.] $\emptype^e \rightarrow \emptype$,
\item[vi.] $\NN^e \rightarrow \NN$,
\item[vii.] For any $a,b:A$, we have $(a\exoeq b) \rightarrow (a = b)$.
\end{itemize}
\end{thm}

\begin{proof}
For each one, the definition follows from the elimination rule of the corresponding exo-types. 

\textit{i}. The map $x\mapsto \star$ is an isomorphism with the inverse $x \mapsto \star^e$.

\textit{ii}. The map $\pair^e(a,b) \mapsto \pair (a,b)$ is an isomorphism with the inverse $\pair(a,b) \mapsto \pair^e (a,b)$.

\textit{iii}. The map $f \mapsto (a\mapsto f(a))$ is an isomorphism with the inverse as denoted the same. Note that we do not take the identity map because domain and codomain are not the same.

\textit{iv}. The map is defined as $\inl^e a \mapsto \inl a$ and $\inr^e b \mapsto \inr b$.

\textit{v}. The map is the usual null map that corresponds to the principle \emph{ex falso quodlibet}.

\textit{vi}. The map, say $f$, is defined as $\zero^e \mapsto \zero$ and $\suc^e (n) \mapsto \suc (f(n))$.

\textit{vii}. The map is defined as $\refl^e \mapsto \refl$.
\end{proof}

\begin{remark}\label{A1-axiom}
It is worth emphasizing that the inverses of the maps \textit{iv}, \textit{v}, and \textit{vi} can be assumed to exist. There are some models where these hold (for the details, see the discussion below Lemma 2.11 \cite{2ltt}). However, a possible inverse for the map \textit{vii} would yield a contradiction because the univalence axiom is inconsistent with the uniqueness of identity proofs. This conversion from the exo-equality ($\exoeq$) to the identity ($=$) has still an importance in many proofs later. Thus we denote it by \[\exotoid : {\prod_{a,b:A}}^e(a\exoeq b) \rightarrow (a = b) \,.\] One of its useful corallaries is the following lemma.
\end{remark}

\begin{lemma}\label{isotoeq}
Let $A,B : \U$ be two types. If $A\cong B$, then $A\simeq B$.
\end{lemma}

\begin{proof}
Let $f :A \rightarrow B$ and $g:B \rightarrow A$ be such that \[p:f \ecirc g \exoeq \id_B\quad \text{and} \quad q:g \ecirc f \exoeq \id_A.\] Since, $A\rightarrow B$ and $B\rightarrow A$ are also (function) types, we get \[\exotoid(p):f \ecirc g = \id_B\quad \text{and} \quad \exotoid(q):g \ecirc f = \id_A.\] Since $f\ecirc g = f \circ g$ and $g\ecirc f = g\circ f$ hold by definition, we are done. 
\end{proof}

\textbf{Agda Side.} The file \texttt{C.agda} contains the complete proof of Lemma \ref{lemma211}. In addition to the flag \texttt{--two-level}, we should also use another flag, which is \texttt{--cumulativity}. This enables the subtyping rule \texttt{Set i $\leq$ SSet i} for any level \texttt{i}. Thanks to the flag, we can take types as arguments for the operations/formations that are originally defined for exo-types, as we discussed at the beginning of this section. There is a problem with the usage of these two flags together, and Figure \ref{agdabug} gives an example of this.
\begin{figure}[ht]
    \centering
        \includegraphics[scale=0.5]{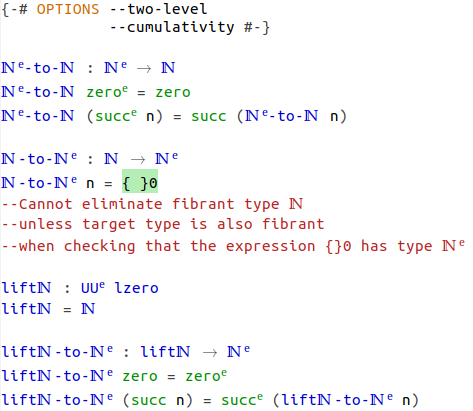}
    \caption{An example of the interaction between the flags \texttt{--two-level} and \texttt{--cumulativity}.}
    \label{agdabug}
\end{figure}

As in the example, Agda has a feature that prevents elimination from a fibrant type to a non-fibrant type. However, when we lift $\NN$ to a term in $\U^e$, which is obtained by cumulativity, it is not a fibrant type anymore, and Agda allows the elimination of its constructors. As noted in Remark \ref{A1-axiom}, such maps are not entirely wrong and can be added to our theory. However, there are essential models where they do not exist, so they should not be definable in an implementation. Also, there are other problems emerging due to the cumulativity itself\footnote{See \url{https://github.com/agda/agda/issues/5761}}.

\subsection{Fibrant exo-types}

\begin{defn}
An exo-type $A : \U^e$ is called a \textbf{fibrant exo-type} if there is a type $RA : \U$ such that $A$ and $RA$ are exo-isomorphic. In other words, $A$ is fibrant when the following exo-type is inhabited \[\isFibrant(A):={\sum_{RA : \U}}^e (A \cong RA)\,.\] 
\end{defn}

\begin{prop}[\cite{2ltt}]\label{fibrant}
The following are true:
\begin{itemize}
\item[i.] Any type $A : \U$ is a fibrant exo-type.
\item[ii.] The unit exo-type $\unit^e$ is fibrant.
\item[iii.] Let $A : \U^e$ and $B : A \rightarrow \U^e$, if $A$ is fibrant, and each $B(a)$ is fibrant, then both $\sum^e_{a:A} B(a)$ and $\prod^e_{a:A} B(a)$ are fibrant. 
\item[iv.] If $A,B : \U^e$ are exo-isomorphic types, then $A$ is fibrant if and only if $B$ is fibrant.
\item[v.] If $A:\U^e$ is fibrant, and there are two types $B,C :\U$ such that $A\cong B$ and $A\cong C$, then $B=C$.
\end{itemize}
\end{prop}

\begin{proof}
\begin{itemize}
\item[i.] This is trivial because we can take $(A,\id_A) : \isFibrant(A)$.
\item[ii.] By Theorem \ref{lemma211}, we know that there is an exo-isomorphism $e : \unit^e \cong \unit$.
\item[iii.] Let $RA:\U$ and $RB : A \rightarrow \U$ such that $A \cong RA$ and $B(a) \cong RB(a)$ for each $a:A$. Take $r_A:A\rightarrow RA$ with inverse $s_A:RA \rightarrow A$. Using the functoriality of $\sum$-exo-types and the map \textit{ii} in Theorem \ref{lemma211}, we have \[{\sum_{a:A}}^e B(a) \cong {\sum_{c:RA}}^e RB(s_A(c))\cong \sum_{c:RA} RB(s_A(c)).\] 
Similarly, the functoriality of $\prod$-exo-types and the map \textit{iii} in Theorem \ref{lemma211} imply that
\[{\prod_{a:A}}^e B(a) \cong {\prod_{c:RA}}^e RB(s_A(c)) \cong \prod_{c:RA} RB(s_A(c)).\]
\item[iv.] If $A$ is fibrant, namely, $A\cong RA$ for a type $RA:\U$, since $\cong$ is transitive, we get $B\cong A\cong RA$. The reverse is the same.
\item[v.] By transitivity, we have $B\cong C$. Lemma \ref{isotoeq} implies $A\simeq B$, and the result follows from the univalence. \qedhere
\end{itemize}
\end{proof}

Just as we have fibrant exo-types, we can consider the maps between fibrant exo-types as the maps between types.  

\begin{defn}\label{exoequiv}
Let $A,B: \U^e$ be two fibrant exo-types, and $f:A\rightarrow B$. Let $RA,RB:\U$ be such that $A\cong RA$ and $B\cong RB$. We have the following diagram. 
\begin{center}
\begin{tikzcd}
A \arrow[rr, "f"] \arrow[dd, "r_A"', bend right] &  & B \arrow[dd, "r_B"', bend right]  \\
\cong                                            &  & \cong                             \\
RA \arrow[uu, "s_A"', bend right]                &  & RB \arrow[uu, "s_B"', bend right]
\end{tikzcd}
\end{center}
We call $f$ a \textbf{fibrant-equivalence} if \[r_B\ecirc f \ecirc s_A : RA \rightarrow RB\] is an equivalence.
\end{defn}

\begin{prop}\label{exoequiv-prop}
Let $A,B,C: \U^e$ be fibrant exo-types, and $f:A\rightarrow B$ and $g:B \rightarrow C$. Consider the corresponded diagram. 
\begin{center}
\begin{tikzcd}
A \arrow[rr, "f", shift left] \arrow[dd, "r_A"', bend right] \arrow[rr, "f'"', shift right] &  & B \arrow[dd, "r_B"', bend right] \arrow[rr, "g"] &  & C \arrow[dd, "r_C"', bend right]  \\
\cong                                                                                       &  & \cong                                            &  & \cong                             \\
RA \arrow[uu, "s_A"', bend right]                                                           &  & RB \arrow[uu, "s_B"', bend right]                &  & RC \arrow[uu, "s_C"', bend right]
\end{tikzcd}
\end{center}
The following are true:
\begin{itemize}
\item[i.] If $f$ is an exo-isomorphism, then it is a fibrant-equivalence.
\item[ii.] If $f,g$ are fibrant-equivalences, then so is $g\ecirc f$.
\item[iii.] If $f,f'$ are homotopic with respect to $\exoeq$, and $f$ is a fibrant-equivalence, then $f'$ is a fibrant-equivalence, too.
\end{itemize}
\end{prop}

\begin{proof}
i. Let $h:B\rightarrow A$ be such that $f\ecirc h \exoeq \id_B$ and $h\ecirc f \exoeq \id_A$. We claim that $r_B \ecirc f \ecirc s_A: RA \rightarrow RB$ is an equivalence with the inverse $r_A \ecirc h \ecirc s_B : RB \rightarrow RA$. It can be shown via the following chain of exo-equalities:
\begin{eqnarray*}
(r_B \ecirc f \ecirc s_A) \ecirc (r_A \ecirc h \ecirc s_B) &\exoeq & r_B \ecirc f \ecirc (s_A \ecirc r_A) \ecirc h \ecirc s_B\\
&\exoeq & r_B \ecirc (f \ecirc h) \ecirc s_B \\
&\exoeq & r_B \ecirc s_B\\
&\exoeq & \id_{RB}
\end{eqnarray*}
Then taking $\exotoid$ of this exo-equality proves that \[(r_B \ecirc f \ecirc s_A) \circ (r_A \ecirc h \ecirc s_B) = \id_{RB}.\] The right-inverse identity can be proved similarly.

ii. We need to prove that $r_C \ecirc (g \ecirc f) \ecirc s_A$ is an equivalence. This map is path-homotopic to $(r_C \ecirc g \ecirc s_B) \circ (r_B \ecirc f \ecirc s_A)$. Since compositions and homotopies preserve equivalences, we are done.

iii. By assumption $r_B \ecirc f \ecirc s_A$ is an equivalence. For all $c:RA$, we have
\[r_B(f(s_A(c))) \exoeq r_B(f'(s_A(c)))\] because $f(x)\exoeq f'(x)$ for all $x:A$. Since homotopies preserve equivalences, we get $r_B \ecirc f' \ecirc s_A$ is an equivalence.
\end{proof}

We also have other useful properties for fibrant-equivalences. For example, consider the following diagram between fibrant exo-types.
\begin{center}
\begin{tikzcd}
A \arrow[rr, "f"] \arrow[rrdd, "f\ecirc g"'] &  & B \arrow[dd, "g"] \\
                                              &  &                   \\
                                              &  & C                
\end{tikzcd}
\end{center}
Then \textbf{2-out-of-3 property} says if two of the three maps $f$, $g$ and the composite $g\ecirc f$ are fibrant-equivalences, then so is the third. For another example, consider the commutative diagram between fibrant exo-types.
\begin{center}
\begin{tikzcd}
A \arrow[rr, "f"] \arrow[dd, "g"'] &  & B \arrow[dd, "g'"] \\
                                   &  &                    \\
A' \arrow[rr, "f'"']               &  & B'                
\end{tikzcd}
\end{center}
Then \textbf{3-out-of-4 property} says if three of the four maps $f$, $f'$, $g$, $g'$ are exo-equivalences, then so is the fourth. 

\textbf{Agda Side.} The folder \texttt{Coercion} in the library \cite{agdalib} contains the formalizations of all definitions and propositions in this section. Note that in \cite{UPpaper}, when an exo-type $A$ is fibrant, its fibrant match is also assumed to be $A$ because this paper uses $\mathtt{isFibrant}(A):=\sum^e_{RA:\U} (A \exoeq RA)$ as the definition of fibrancy. If we assume axiom T3\footnote{Axiom T3 says that if an exotype $A$ is isomorphic to a type $B$, then $A$ is itself a (fibrant) type \cite{2ltt}.}, these two are logically equivalent. The main purpose of our choices becomes clearer in the formalization. In the case $A\exoeq B$ where $A$ is an exo-type, and $B$ is a type, we can still define an isomorphism between them by transporting their terms under this equality, but if we have $A\cong B$, we get the isomorphism directly. In this way, we gain practical advantages. For example, the notion of fibrant-equivalence is new, and it is designed for this practical purposes. Also, with this approach, all proofs provided here become the same as in their formalizations. Thus, there is no gap between language and symbolism.

\subsection{Cofibrant exo-types}

In this section, a weaker definition than fibrancy is given. We also provide a new (but logically equivalent) characterization of it. Note that if a property, that is defined for types initially, is attributed to a fibrant exo-type, we emphasize that the property belongs to the fibrant match of the exo-type.

\begin{defn}[\cite{2ltt} Corollary 3.19(i)]\label{cofib-defn} Let $A:\U^e$ be an exo-type. We call it \textbf{cofibrant} if the following holds 
\begin{itemize}
\item For any type family $Y : A \rightarrow \U$ over $A$, the exo-type $\prod^e_{a:A} Y(a)$ is fibrant,
\item In the case above, if $Y(a)$ is contractible for each $a:A$, then so is the fibrant match of $\prod^e_{a:A} Y(a)$.
\end{itemize}
\end{defn}

The following gives a logically equivalent definition of cofibrancy. Attention should be paid to the use of $=$ and $\exoeq$ in order to indicate whether the terms belong to the type or the exo-type.

\begin{prop}\label{funext-cofib}
Let $A:\U^e$ be an exo-type such that for any type family $Y : A \rightarrow \U$ over $A$, the exo-type $\prod^e_{a:A} Y(a)$ is fibrant. Then the following are equivalent:
\begin{itemize}
\item[i.] In the case above, if $Y(a)$ is contractible for each $a:A$, then so is its fibrant match of $\prod^e_{a:A} Y(a)$, namely $A$ is cofibrant.
\item[ii.]{\textbf{(Funext for cofibrant types).}} In the case above, for any $f,g : \prod^e_{a:A} Y(a)$ if $f(a)=g(a)$ for each $a:A$, then $r(f)=r(g)$ where $FM : \U$ and
\begin{center}
\begin{tikzcd}
\prod^e_{a:A} Y(a) \arrow[rr, "r", bend left] & \cong & FM \arrow[ll, "s", bend left]
\end{tikzcd}.
\end{center}
\end{itemize}
\end{prop}

\begin{proof}
(i $\Rightarrow$ ii) Let $f,g : \prod^e_{a:A} Y(a)$ be such that $t_a : f(a)=g(a)$ for each $a:A$. Consider another type family $Y' : A \rightarrow \U$ defined as \[Y'(a):=\sum_{b:B(a)} b=f(a).\] Then both $f':=\lambda a.(f(a),\refl)$ and $g':=\lambda a. (g(a),t_a^{-1})$ are terms in $\prod^e_{a:A}Y'(a)$. By our assumptions, there is a $FM':\U$ such that
\begin{center}
\begin{tikzcd}
\prod^e_{a:A} Y'(a) \arrow[rr, "r'", bend left] & \cong & FM' \arrow[ll, "s'", bend left]
\end{tikzcd}.
\end{center}
Since the type of paths at a point is contractible (Lemma 3.11.8 in \cite{hott}), we have each $Y'(a)$ is contractible. By the assumption (i), we get $FM'$ is contractible, and hence 
\begin{equation}\label{eqprop3}
r'(f)=r'(g).
\end{equation}
Using this, we have the following chain of identities:
\[ r(f)=r(\pi_1(s'(r'(f'))))=r(\pi_1(s'(r'(g'))))=r(g). \] 
The first (and the third, by symmetry) identity is obtained as follows: Because $r'$ and $s'$ are exo-inverses of each other, we have \[(a:A)\rightarrow \pi_1(f'(a))\exoeq \pi_1 ((s'(r'(f'))) (a)).\] Thus, by $\funext^e$, we get $\pi_1(f')\exoeq \pi_1 ((s'(r'(f'))))$. Then we apply\footnote{We mean the usual $\ap$ operation in HoTT by saying ``applying a function to an identity (or exo-equality)". } $r$ to the equality, and make it an identity via $\exotoid$ because the terms are in $FM'$ which is a type. 

The second identity is obtained by the applying to the function \[\lambda x. r (\lambda a. \pi_1((s' x) (a))) : FM' \rightarrow FM\] to the identity \ref{eqprop3}, so we are done. 

Note that even if $r$ has exo-type domain and $s'$ has exo-type codomain, we can compose these in a way that the resulting map is from a type to a type. Thus, we can apply it to an identity.

(ii $\Rightarrow$ i) Suppose $Y(a)$ is contractible for each $a:A$. We want to show that $FM$ is contractible where
\begin{center}
\begin{tikzcd}
\prod^e_{a:A} Y(a) \arrow[rr, "r", bend left] & \cong & FM \arrow[ll, "s", bend left]
\end{tikzcd}.
\end{center}

Let $b_a :Y(a)$ be the center of contraction for each $a:A$. Then this gives a function $f:=\lambda a. b_a :\prod^e_{a:A} Y(a)$. For any $x:FM$, since $f(a)=b_a=s(x)(a)$ for each $a:A$ by contractibility asumption, we get $r(f)=r(s(x))$. Also, applying $\exotoid$ to the exo-equality $r(s(x))\exoeq x$, we get $r(s(x))=x$, and transitivity of $=$ yields $r(f)=x$. Therefore, $FM$ is contactible with the center of contraction $r(f)$.
\end{proof}

Cofibrant exo-types have the following properties.
\begin{prop}[\cite{2ltt}]\label{cofibrant-prop} 
\item[i.] All fibrant exo-types are cofibrant.
\item[ii.] If $A$ and $B$ are exo-types such that $A\cong B$, and if $A$ is cofibrant, then $B$ is cofibrant. 
\item[iii.] $\emptype^e$ is cofibrant, and if $A,B:\U^e$ are cofibrant, then so are $A+^eB$ and $A\times^eB$. In particular, all exo-finite exo-types are cofibrant.
\item[iv.] If $A:\U^e$ is cofibrant and $B:A\rightarrow\U^e$  is such that each $B(a)$ is cofibrant, then $\sum^e_{a:A} B(a)$ is cofibrant.
\end{prop}

\begin{proof}
i. Let $A$ be a fibrant exo-type, say \begin{tikzcd}
A \arrow[r, "r", shift left] & RA \arrow[l, "s", shift left]
\end{tikzcd}. Then for any $Y:A\rightarrow \U$, by Proposition \ref{fibrant}, we know that $\prod^e_{a:A} Y(a)$ is fibrant with the fibrant match $\prod_{c:RA} Y(s(c))$. If each $Y(a)$ is contractible, then so is $\prod_{c:RA} Y(s(c))$ because this holds for ordinary types. 

ii. Let $f:A\rightarrow B$ be an exo-isomorphism with the inverse $g:B \rightarrow A$. Let $Y:B \rightarrow \U$. Consider the following diagram.
\begin{center}
\begin{tikzcd}
\prod^e_{b:B} Y(b) \arrow[rr, "u", shift left] &  & \prod^e_{a:A} Y(f(a)) \arrow[ll, "v", shift left] \arrow[rr, "r_A", shift left] &  & FM_A \arrow[ll, "s_A", shift left]
\end{tikzcd}
\end{center}
The maps $u,v$ form isomorphisms by the $\prod$-functoriality property. The isomorphisms $r_A,s_A$ are obtained by the cofibrancy of $A$. Therefore, we get $\prod^e_{b:B} Y(b)$ is fibrant. Also, if each $Y(b)$ is contractible, then in particular $Y(f(a))$ is contractible for each $a:A$. So $FM_A$ is contractible since $A$ is cofibrant. This proves that $B$ is cofibrant. 

iii. In the case of $\emptype^e$, for any $Y:\emptype^e \rightarrow \U$, it's easy to show that $\prod^e_{c:\emptype^e}Y(c)\cong \unit$. This satisfies two conditions necessary for being cofibrant simultaneously. 

For an exo-coproduct, if $A,B:\U^e$ are cofibrant, and $Y:A+^e B \rightarrow \U$, we have the following diagram. 
\begin{center}
\begin{tikzcd}
\prod^e_{c:A+^e B} Y(c) \arrow[rrr, "u", shift left] \arrow[dd, no head] \arrow[dd, no head, shift right] &  &  & (\prod^e_{a:A} Y(\inl(a)))\times^e (\prod^e_{b:B} Y(\inr(b))) \arrow[lll, "v", shift left] \arrow[dd, "r_A\times r_B"', bend right] \\
                                                                                                            &  &  &                                                                                                                                       \\
\prod^e_{c:A+^e B} Y(c)                                                                                   &  &  & FM_A \times FM_B \arrow[uu, "s_A \times s_B"', bend right]                                                                           
\end{tikzcd}
\end{center}

Here, the maps $u,v$ have their own obvious definitions, and they are isomorphisms. (Actually, the analogous statement is true for types\footnote{See Exercise 2.9 in the HoTT Book \cite{hott}}, and the exo-type version can be proven similarly in terms of $\exoeq$.) The exo-isomorphisms $r_A$ and $r_B$ come from cofibrancy of $A$ and $B$, respectively. It is easy to see that $r_A \times r_B$ is an exo-isomorphism. Therefore, $(r_A \times r_B) \circ^e u$ is the exo-isomorphism we searched for. If each $Y(c)$ is contractible, in particular, both $Y(\inl(a))$ and $Y(\inr(b))$ are contractible. By the assumption, this means that both $FM_A$ and $FM_B$ are contractible, and so is $FM_A \times FM_B$. This proves that $A+^e B$ is cofibrant.

The case of an exo-product is a particular case of $\sum^e$-exo-types. For exo-finite exo-types, using induction on $\NN^e$, it follows from that both $\emptype^e$ and $\unit^e$ are cofibrant, and cofibrancy is preserved under exo-coproducts. 

iv. Suppose $A$ is cofibrant and each $B(a)$ is cofibrant. Let $Y: \sum^e_{a:A} B(a) \rightarrow \U$ and consider the following diagram. 

\begin{tikzcd}
\prod^e_{c:\sum^e_{a:A} B(a)} Y(c) \arrow[rr, "u", shift left=2] &  & {\prod^e_{a:A} \prod^e_{b:B(a)} Y(a,b)} \arrow[ll, "v"] \arrow[dd, "r_{B(a)}"', shift right=2] &  &                                    \\
                                                                 &  &                                                                                                &  &                                    \\
                                                                 &  & \prod^e_{a:A} FM_{B(a)} \arrow[uu, "s_{B(a)}"'] \arrow[rr, "r_A", shift left]                  &  & FM_A \arrow[ll, "s_A", shift left]
\end{tikzcd}

The maps $u,v$ form the usual isomorphism by the expansion of $\sum^e$-exo-type. (Actually, the same is true for types, and the exo-type version can also be proven similarly in terms of $\exoeq$.) The maps $r_{B(a)}$ and $s_{B(a)}$ form an isomorphism by the cofibrancy of $B(a)$ and the $\prod$-functoriality property. The last pair of maps $r_A,s_A$ form an isomorphism by the cofibrancy of $A$. Therefore, the exo-type in the top left corner is fibrant. Furthermore, if each $Y(c)$ is contractible, in particular $Y(a,b)$ is contractible, then each $FM_{B(a)}$ is contractible by the cofibrancy of $B(a)$. This proves that $FM_A$ is contractible by the cofibrancy of $A$. This proves that $\sum^e_{a:A} B(a)$ is cofibrant.
\end{proof}

\textbf{Agda Side.} The folder \texttt{Cofibration} in the library \cite{agdalib} is about the material in this section. The file \verb|Funext_for_cofibrant_types.agda| is Proposition \ref{funext-cofib}. This is one of the major contributions of Agda formalization to the theory. We need it while we try to formalize some proofs in the next section. We already know some functorial properties of dependent function types in the usual HoTT. When we have dependent exo-types $A : B \rightarrow \U$ for $A:\U^e$, the exo-type $\prod_{a:A} B(a)$ can be fibrant or not. However, we still need the functoriality rules for it, which are very useful in the applications\footnote{For example, see Figure \ref{sharp-sigma}.}. Then Proposition \ref{funext-cofib} addresses this need by giving another version of \texttt{funext}, and it works well enough.

\subsection{Sharp exo-types}

Another class of exo-types is the class of sharp ones. This was given in \cite{UPpaper} for the first time. As in the previous section, we'll give the proofs without any abuse of notation, and so their formalization is the same as in this paper. 

\begin{defn}[Def. 2.2 \cite{UPpaper}]
An exo-type $A$ is \textbf{sharp} if it is cofibrant, and and it has a “fibrant replacement”, meaning that there is a fibrant type $RA$ and a map $r : A \rightarrow RA$ such that for any family of types $Y : RA \rightarrow \U$, the precomposition map below 
is a fibrant-equivalence (recall Definition \ref{exoequiv}).
\begin{equation}\label{precomp}
(- \ecirc r):\prod_{c:RA} Y(c) \rightarrow {\prod_{a:A}}^e Y(r(a))
\end{equation}
\end{defn}

The following lemma gives another definition for sharp exo-types. First, we need an auxiliary definition.

\begin{defn}
Let $A,B: \U^e$ be two fibrant exo-types, and $f:A\rightarrow B$ be a map. Let $RA,RB:\U$ be such that $A\cong RA$ and $B\cong RB$. Take $s_A:RA \rightarrow A$ and $r_B:B \rightarrow RB$ as these isomorphisms. Then $f$ has a \textbf{fibrant-section} if \[r_B\ecirc f \ecirc s_A : RA \rightarrow RB\] has a section, namely, there is $g : RB \rightarrow RA$ such that $(r_B\ecirc f \ecirc s_A) \ecirc g = \id_{RB}$.
\end{defn}

\begin{lemma}[\cite{UPpaper}]\label{equiv-defsharp}
Let $A$ be a cofibrant exo-type, $RA$ a type, and $r:A \rightarrow RA$ a map. The following are equivalent:
\begin{enumerate}
\item The map (\ref{precomp}) is a fibrant-equivalence for any $Y:RA\rightarrow \U$, so that $A$ is sharp.
\item The map (\ref{precomp}) has a fibrant-section for any $Y:RA\rightarrow \U$.
\item The map (\ref{precomp}) is a fibrant-equivalence whenever $Y:=\lambda x. Z$ for a constant type $Z:\U$, hence $RA \rightarrow Z$ is equivalent to the fibrant match of $A \rightarrow Z$.
\end{enumerate}
\end{lemma}

\begin{proof}
$(1 \Rightarrow 2)$ follows from the fact that an equivalence has a section. $(1 \Rightarrow 3)$ is trivial because 3 is a particular case of 1. 

$(2 \Rightarrow 1)$ Let $Y:RA\rightarrow \U$ be a family of types over $A$ and consider the relevant diagram. 
\begin{center}
\begin{tikzcd}
\prod_{c:RA} Y(c) \arrow[rr, "(-\ecirc r)"] \arrow[dd, no head] \arrow[dd, no head, shift left] &  & \prod_{a:A} Y(r(a)) \arrow[dd, "r_A"', shift right] \\
                                                                                                 &  &                                                     \\
\prod_{c:RA} Y(c) \arrow[rr, "\alpha:=r_A \ecirc (-\ecirc r)"]                                 &  & FM \arrow[uu, "s_A"', shift right]                 
\end{tikzcd}
\end{center}
By the assumption, we know that $\alpha$ has a section, say $\beta : FM \rightarrow \prod_{c:RA} Y(c)$, so that we have $\alpha \circ \beta = id_{FM}$. It is enough to show that $\beta \circ \alpha = \id_{\prod_{RA} Y}$ We claim that for any $f,g : \prod_{c:RA} Y(c)$, we have \[\alpha (f) = \alpha (g) \rightarrow f=g.\] If this is true, for any $h:\prod_{c:RA} Y(c)$, taking $f:=\beta(\alpha (h))$ and $g:=h$ we are done by $\funext$. Now, for the proof of the claim, consider another type family over $RA$ defined as $Y'(c):=(f(c)=g(c))$, so that we have another diagram 
\begin{center}
\begin{tikzcd}
\prod_{c:RA} Y'(c) \arrow[rr, "(-\ecirc r)"] \arrow[dd, no head] \arrow[dd, no head, shift left] &  & \prod_{a:A} Y'(r(a)) \arrow[dd, "r'_A"', shift right] \\
                                                                                                 &  &                                                     \\
\prod_{c:RA} Y'(c) \arrow[rr, "\alpha':=r'_A \ecirc (-\ecirc r)"]                                 &  & FM' \arrow[uu, "s'_A"', shift right]                 
\end{tikzcd}
\end{center}
and again $\alpha'$ has a section $\beta'$. If we assume $\alpha (f) = \alpha (g)$ for $f,g : \prod_{c:RA} Y(c)$, this is \[q:r_A (f \ecirc r) = r_A (g \ecirc r)\] by the definition. Define $T:=\lambda (a:A). p_a $ where $p_a:f(r(a))=g(r(a))$ is obtained by 
\begin{eqnarray*}
f(r(a))&=&(s_A(r_A(f \ecirc r))) (a)\\
&=&(s_A(r_A(g \ecirc r))) (a) = g(r(a)).
\end{eqnarray*}
The first and the last ones follow from the fact that $r_A$ and $s_A$ are inverses of each other. The middle identity is obtained by the applying $\lambda u. (s_A(u))(a)$ to the path $q$. Finally, we are done by using $\funext$ on $\beta'(s_A'(T)):\prod^e_{c:RA} f(c)=g(c)$. 

$(3 \Rightarrow 2)$ Let $Y:RA \rightarrow \U$ be a family of types over $RA$ and define $Z:=\sum_{c:RA} Y(c)$. By our assumptions, we have the following diagrams. Our aim is to show that $\alpha$ has a section.
\begin{center}
\begin{tikzcd}
RA \rightarrow Z \arrow[rr, "(-\ecirc r)"] \arrow[dd, no head, shift right] \arrow[dd, no head] &        & A \rightarrow Z \arrow[dd, "r_Z"', shift right=2]                        &            & \prod_{c:RA}Y(c) \arrow[rr, "(-\ecirc r)"] \arrow[dd, no head, shift right] \arrow[dd, no head] &  & {\prod_{a:A}}^e Y(r(a)) \arrow[dd, "r_A"', shift right=2] \\
                                                                                                &        & \cong                                                                    & \text{and} &                                                                                                 &  & \cong                                                     \\
RA \rightarrow Z \arrow[rr, "\alpha_Z:=r_Z\ecirc (-\ecirc r)", shift left=2]                                           & \simeq & FM_Z \arrow[uu, "s_Z"', shift right=2] \arrow[ll, "\beta_Z", shift left=2] &            & \prod_{c:RA}Y(c) \arrow[rr, "\alpha:=r_A\ecirc (-\ecirc r)"]                                                           &  & FM \arrow[uu, "s_A"', shift right=2]                     
\end{tikzcd}
\end{center}
For any $f:\prod^e_{a:A} Y(r(a))$ define $f':A \rightarrow Z$ by $f'(a):=(r(a),f(a))$. For all $a:A$, we have $q_a: \pi_1(\beta_Z(r_Z (f'))\, (r (a)))= r(a)$ obtained by applying $\pi_1$ to
\begin{eqnarray*}
(\beta_Z(r_Z (f')))\, (r (a))&=& s_Z (\alpha_Z (\beta_Z (r_Z(f')))) (a)\\
&=& s_Z (r_Z(f')) (a)\\
&=& f'(a)
\end{eqnarray*}
The first comes from the fact that $u\ecirc r \exoeq (s_Z \ecirc \alpha_Z) (u)$ for all $u:RA \rightarrow Z$. The second and the third are followed by the exo-isomorphisms $r_Z,s_Z$ and the equivalence $\alpha_Z,\beta_Z$. Consider another diagram obtained by the assumption 3 for $Z:=RA$.
\begin{center}
\begin{tikzcd}
RA \rightarrow RA \arrow[dd, no head, shift right] \arrow[dd, no head] \arrow[rr, "(-\ecirc r)"] &        & A \rightarrow RA \arrow[dd, "r_{RA}"', shift right=2]                               \\
                                                                                                 &        & \cong                                                                               \\
RA \rightarrow RA \arrow[rr, "\alpha_{RA}:=r_{RA}\ecirc (-\ecirc r)", shift left=2]                                        & \simeq & FM_{RA} \arrow[uu, "s_{RA}"', shift right=2] \arrow[ll, "\beta_{RA}", shift left=2]
\end{tikzcd}
\end{center}
By $\funext$ for cofibrant exo-types, since \[\lambda a . q_a:\prod^e_{a:A} \pi_1(\beta_Z(r_Z (f'))\, (r (a)))= r(a),\] we get \[r_{RA} (\pi_1 \ecirc (\beta_Z(r_Z (f'))) \ecirc r) = r_{RA} (r).\]
Then for all $f:\prod^e_{a:A} Y(r(a))$, we have $p_f : \pi_1 \ecirc (\beta_Z(r_Z(f'))) = \id_{RA\rightarrow RA}$ obtained by
\begin{eqnarray*}
\pi_1(\beta_Z(r_Z(f')))&=& \beta_A(\alpha_A (\pi_1 \ecirc (\beta_Z(r_Z (f')))))\\
&=& \beta_A(r_{RA} (\pi_1 \ecirc (\beta_Z(r_Z (f'))) \ecirc r))\\
&=& \beta_A(r_{RA}(r))\\
&=& \beta_A(\alpha_A (\id_{RA\rightarrow RA}))\\
&=&\id_{RA\rightarrow RA}.
\end{eqnarray*}
These come from exo-isomorphisms and/or equivalences. Finally, we define the section map $\beta : FM \rightarrow \prod_{c:RA} Y(c)$ by the following\footnote{Here, $\tr$ and $\happly$ are the operations related by the identity types, see \cite{hott} for the details.}: \[\beta(x):= \lambda (c:RA).\, \tr\,\, (\happly p_{s_A(x)} \,\,c)\,\, \pi_2(\beta_Z(r_Z((s_A(x)')))(c)).\] To finish the proof, we should show $\alpha \circ \beta = \id_{FM}$. Let $x:FM$, then for all $a:A$, we have 
\begin{equation}\label{auxeq}
\beta(x)(r(a))= s_A(x)(a).
\end{equation}
This comes from applying $\pi_2$ to the identity $(r(a),\beta(x)(r(a)))=(s_A (x))' (a)$ which is obtained by
\begin{eqnarray*}
(r(a),\beta(x)(r(a)))&=& \beta_Z(r_Z(s_A(x)'))\,(r(a))\\
&=& s_Z (\alpha_Z(\beta_Z(r_Z(s_A(x)'))))\, (a)\\
&=& s_Z(r_Z(s_A(x)'))(a)\\
&=& s_A(x)'(a)
\end{eqnarray*}
Here, the first identity is followed by the lifting property\footnote{Lemma 2.3.2 in the HoTT Book \cite{hott}} of the transport map, the others are obtained by exo-isomorphisms or equivalences. Now, by funext for cofibrant types, the identity (\ref{auxeq}) proves that $r_A(\beta(x)\ecirc r) =r_A(s_A(x))$. However, the left side is equal to $(\alpha \circ \beta) (x)$ by definition, and the right side is equal to $x$ due to the exo-isomorphism. Therefore, $\beta$ is a section for $\alpha$, which proves the statement 2.
\end{proof}

As in the cofibrant exo-types, the notion of sharpness has its own preservence rules. The following proposition gives these rules.

\begin{prop}[\cite{UPpaper}] The following are true:

\begin{itemize}\label{sharp-prop}
\item[i.] All fibrant exo-types are sharp.
\item[ii.] If $A$ and $B$ are exo-types such that $A\cong B$, and if $A$ is sharp, then $B$ is sharp. 
\item[iii.] $\emptype^e$ is sharp, and if $A$ and $B$ are sharp exo-types, then so are $A +^e B$ and $A \times^e B$.
\item[iv.] If $A$ is a sharp exo-type, $B : A \rightarrow \U$ is such that each $B(a)$ is sharp, then $\sum^e_{a:A} B(a)$ is sharp.
\item[v.] Each finite exo-type $\NN^e_{<n}$ is sharp.
\item[vi.] If $\NN^e$ is cofibrant, then it is sharp. 
\end{itemize}
\end{prop}

\begin{proof}
For (i) if $A$ is a fibrant exo-type, and $RA:\U$ is such that $A\cong RA$, then we can take $RA$ as the fibrant replacement. By Proposition \ref{cofibrant-prop}(i) $A$ is cofibrant, and the map (\ref{precomp}) is trivally a fibrant-equivalence. 

By Proposition \ref{cofibrant-prop}(iii), $\emptype^e$ is cofibrant, and we can take $\emptype$ as fibrant replacement.  Also, the $\prod$-type and $\prod^e$-exo-type in the map (\ref{precomp}) are contractible and exo-contractible, respectively. Thus, it is trivially a fibrant-equivalence. 

The statement (ii) can be shown as in Proposition \ref{cofibrant-prop}(ii). Since $\unit^e$ is fibrant, and hence sharp, the sharpness of finite types follows from that sharpness is preserved under an exo-coproduct. The case of an exo-product is a particular case of $\sum^e$-exo-types. Thus, it remains to show the exo-coproduct case, (iv) and (vi).

For an exo-coproduct, let $A$ and $B$ be two sharp exo-types. By Proposition \ref{cofibrant-prop}(iii), $A+^eB$ is cofibrant. Let $RA,RB:\U$ be the fibrant replacements of $A,B$, respectively. Let $r_A:A\rightarrow RA$ and $r_B:B\rightarrow RB$ be the relevant maps. We claim that $RA + RB$ is a fibrant replacement of $A+^eB$ with the map $r:A+^eB \rightarrow RA + RB$ defined by
\[ r (x) :\equiv
  \begin{cases}
    \inl r_A (a) & \text{ if } x=\inl^e a, \\
    \inr r_B (b) & \text{ if } x=\inr^e b.
  \end{cases} 
  \]
For $Y : RA + RB \rightarrow \U$ consider the commutative diagram in Figure \ref{sharp-exocoprod}. In the diagram, the equivalences $\simeq_1$ and $\simeq_2$, and the isomorphism $\cong_1$ follow from the universal property of (exo)coproducts \cite{hott}. The three pairs of maps $(u_A,v_A)$, $(u_B,v_B)$, and $(u,v)$ are obtained by cofibrancy of $A$, $B$, and $A+^eB$, and hence these are all exo-isomorphisms. Since $A$ and $B$ are sharp, the map $\alpha$ and $\beta$ are equivalences. It is then easy to see that $\alpha \times \beta$ is also an equivalence. The compostion of three arrows on the left is the precomposition map $(-\ecirc r)$. Observe that the composition of three arrows on the right is an equivalence. Indeed, the first and second are already equivalences. Since we have the isomorphism $\cong_2$ between types $FA\times FB$ and $FM$, this is also an equivalence by Proposition \ref{exoequiv-prop}(i). Therefore, the right composition of arrows is an equivalence, so $(-\ecirc r)$ is a fibrant-equivalence.

\begin{figure}
\begin{tikzcd}
\prod_{c:RA+RB} Y(c) \arrow[rr, no head, shift right] \arrow[rr, no head] \arrow[dd, "\simeq_1"']                                                            &  & \prod_{c:RA+RB} Y(c) \arrow[dd, "\simeq_2"]                                            \\
                                                                                                                                                             &  &                                                                                        \\
\prod_{x:RA} Y(\inl x) \times \prod_{y:RB} Y(\inr y) \arrow[rr, no head] \arrow[rr, no head, shift right] \arrow[dd, "(-\ecirc r_A)\times^e (-\ecirc r_B)"'] &  & \prod_{x:RA} Y(\inl x) \times \prod_{y:RB} Y(\inr y) \arrow[dd, "\alpha \times \beta"] \\
                                                                                                                                                             &  &                                                                                        \\
\prod^e_{a:A}Y(\inl r_A(a)) \times^e \prod^e_{b:B} Y(\inr r_B(b)) \arrow[rr, "u_A \times^e u_B"', shift right=2] \arrow[dd, "\cong_1"']                      &  & FA \times FB \arrow[ll, "v_A \times^e v_B"', shift right] \arrow[dd, "\cong_2"]        \\
                                                                                                                                                             &  &                                                                                        \\
\prod^e_{d:A+^e B}Y(r(d)) \arrow[rr, "u"', shift right=2]                                                                                                    &  & FM \arrow[ll, "v"', shift right]                                                      
\end{tikzcd}
\caption{The diagram about sharpness of exo-coproduct.\label{sharp-exocoprod}}
\end{figure}

For a dependent sum, let $A$ be a sharp exo-type, $B:A\rightarrow \U$ be such that each $B(a)$ is sharp. By Proposition \ref{cofibrant-prop}(iv), we have $\sum^e_{a:A}B(a)$ is cofibrant. It remains to find a fibrant replacement. 

Let $r_A:A\rightarrow RA$ be the fibrant replacement of $A$, and $r_a : B(a) \rightarrow RB(a)$ be the fibrant replacement of each $B(a)$ for $a:A$. We have $RB : A \rightarrow \U$. Consider the diagram obtained by the sharpness of $A$. 
\begin{center}
\begin{tikzcd}
RA \rightarrow \U \arrow[rr, "(-\ecirc r_A)"] \arrow[dd, no head, shift right] \arrow[dd, no head] &        & A \rightarrow \U \arrow[dd, "u"', shift right=2]                    \\
                                                                                                    &        & \cong                                                                \\
RA\rightarrow \U \arrow[rr, "\alpha", shift left=2]                                                & \simeq & FM \arrow[uu, "v"', shift right=2] \arrow[ll, "\beta", shift left=2]
\end{tikzcd}
\end{center}
Define $\widetilde{RB}:=\beta (u (RB)):RA \rightarrow \U$, and we will make $\sum_{c:RA} \widetilde{RB}(a)$ the fibrant replacement of $\sum^e_{a:A} B(a)$. Define $r : \sum^e_{a:A} B(a) \rightarrow \sum_{c:RA} \widetilde{RB}(a)$ by \[r(a,b):=(\,r_A(a)\,,\,e_a(r_a(b))\,) \] where $e:\prod^e_{a:A}RB(a)\simeq \widetilde{RB}(r_A (a))$. The family of equivalences $e$ is obtained by the identity, for all $a:A$ 
\begin{eqnarray*}
\widetilde{RB}(r_A(a))=v (u (\widetilde{RB} \ecirc r_A)) (a)&=&v(\alpha(\widetilde{RB}))(a)\\
&=&v(\alpha(\beta(u(RB))))(a)=v(u(RB))(a)=RB(a).
\end{eqnarray*}
It remains to show that $(-\ecirc r)$ is a fibrant-equivalence. Consider the commutative diagram in Figure \ref{sharp-sigma} for the dependent type $Y: \sum_{c:RA} \widetilde{RB}(a) \rightarrow \U$.

\begin{figure}
\centering
\begin{tikzcd}
\prod_{z:\sum_{c:RA} \widetilde{RB}(c)} Y(z) \arrow[dd, "\phi_0"'] \arrow[rr, no head, shift left] \arrow[rr, no head]      &       & \prod_{z:\sum_{c:RA} \widetilde{RB}(c)} Y(z) \arrow[dd, "\phi_0"]                                                 \\
                                                                                                                            &       &                                                                                                                   \\
{\prod_{c:RA} \prod_{y:\widetilde{RB}(c)} Y(c,y)} \arrow[dd, "\phi_1"'] \arrow[rr, no head, shift left] \arrow[rr, no head] &       & {\prod_{c:RA} \prod_{y:\widetilde{RB}(c)} Y(c,y)} \arrow[dd, "\alpha_1"', shift right=2]                          \\
                                                                                                                            &       & \simeq                                                                                                            \\
{\prod^e_{a:A} \prod_{y:\widetilde{RB}(r_A(a))} Y(r_A(a),y)} \arrow[dd, "\phi_2"'] \arrow[rr, "u_1", shift left=2]          & \cong & FM_1 \arrow[ll, "v_1", shift left=2] \arrow[uu, "\beta_1"', shift right=2] \arrow[dd, "\alpha_2"', shift right=2] \\
                                                                                                                            &       & \simeq                                                                                                            \\
{\prod^e_{a:A}\prod_{x:RB(a)} Y(r_A(a),e_a(x))} \arrow[dd, "\phi_3"'] \arrow[rr, "u_2", shift left=2]                       & \cong & FM_2 \arrow[ll, "v_2", shift left=2] \arrow[uu, "\beta_2"', shift right=2] \arrow[dd, "\alpha_3"', shift right=2] \\
                                                                                                                            &       & \simeq                                                                                                            \\
{\prod^e_{a:A} \prod^e_{b:B(a)} Y(r_A(a),e_a(r_a(b)))} \arrow[dd, "\phi_4"'] \arrow[rr, "u_3", shift left=2]                & \cong & FM_3 \arrow[ll, "v_4", shift left=2] \arrow[uu, "\beta_3"', shift right=2] \arrow[dd, "\alpha_4"', shift right=2] \\
                                                                                                                            &       & \simeq                                                                                                            \\
\prod^e_{s:\sum^e_{a:A} B(a)} Y(r(s)) \arrow[rr, "u_4", shift left=2]                                                       & \cong & FM_4 \arrow[ll, "v_4", shift left=2] \arrow[uu, "\beta_4"', shift right=2]                                       
\end{tikzcd}
\caption{The diagram about sharpness of $\sum$-exo-type.\label{sharp-sigma}}
\end{figure}

In Figure \ref{sharp-sigma}, the first and the second rows contain two types, and there is a canonical equivalence $\phi_0$ between them obtained by a universal property. Similarly, the map $\phi_4$ is a canonical isomorphism. Also, the maps $\phi_1$, $\phi_2$, and $\phi_3$ have their own obvious definitions.  The types from $FM_1$ to $FM_4$ are obtained by the cofibrancy of $A$, $B(a)$ and $\sum^e_{a:A} B(a)$. Therefore, the pairs of maps from $(u_1,v_1)$ to $(u_4,v_4)$ are isomorphisms. 

The pair $(\alpha_1,\beta_1)$ is an equivalence since $A$ is sharp. The pair $(\alpha_2,\beta_2)$ is an equivalence by a functoriality rule\footnote{If $A\simeq B$ and $P:A \rightarrow \U$, $Q:B \rightarrow \U$ such that $P\simeq Q$, then $\prod_A P \simeq \prod_B Q$. We have similar rule for exo-types with $\cong$.} since $RB(a)\simeq \widetilde{RB}(r_A (a))$ for any $a:A$. The pair $(\alpha_3,\beta_3)$ is an equivalence by a similar functoriality rule and the fact that $B(a)$ is sharp for each $a:A$. The pair $(\alpha_4,\beta_4)$ is an isomorphism since the other sides of the last square are all isomorphisms, and hence it is an equivalence. 

Now, while the composition of the maps on the left is the precomposition $(-\ecirc r)$, the composition of the maps on the right is an equivalence. Thus, it finishes the proof that $\sum^e_{a:A} B(a)$ is sharp.

As for the statement (vi), suppose $\NN^e$ is cofibrant. We take $\NN$ as a fibrant replacement, and the transfer map $r:\NN^e \rightarrow \NN$ defined by $r(\zero^e):=\zero$ and $r(\suc^e n)=\suc r(n)$. For any $Y:\NN \rightarrow \U$,  consider the diagram:
\begin{center}
\begin{tikzcd}
\prod_{n:\NN} Y(n) \arrow[rr, "(-\ecirc r)"] \arrow[dd, no head, shift right] \arrow[dd, no head] &  & \prod^e_{m:\NN^e} Y(r(m)) \arrow[dd, "u_Y"', shift right=2] \\
                                                                                                  &  &                                                           \\
\prod_{n:\NN} Y(n) \arrow[rr, "\alpha_Y"]                                                           &  & FM_Y \arrow[uu, "v_Y"', shift right=2]                       
\end{tikzcd}
\end{center}

Using Lemma \ref{equiv-defsharp}, we will show that $(-\ecirc r)$ has a fibrant-section, namely, $\alpha_Y=u_Y \ecirc (-\ecirc r)$ has a section. First, we define an auxiliary type \[S : \prod_{n:\NN} \,\, \prod_{Y:\NN \rightarrow \U} \,\, \prod_{x:FM_Y} Y(n)\] by $S(\zero,Y,x):= v_Y (x) (\zero^e)$ and $S(\suc(n),Y,x):=S(n,Y',x')$ where
\begin{eqnarray*}
Y'(n)&:=&Y(\suc(n)),\\
x'&:=& u_{Y'}(\lambda a.\,v_Y(x)\,(\suc^e (a))\,).
\end{eqnarray*}

We then define the section map $\beta_Y : FM_Y \rightarrow \prod_{n:\NN} Y(n)$ as \[\beta_Y(x)(n):=S(n,Y,x).\] To finish the proof, it remains to show that $\alpha_Y\circ \beta_Y = \id$. By $\funext$, it suffices to show that for any $x:FM_Y$, we have $\alpha_Y(\beta_Y(x))=x$. Using $\funext$ for cofibrant types, it is enough to show that for any $m:\NN^e$, we have 
\begin{equation}\label{aux-eq-nat}
(\beta_Y(x))\,(r(m)) = v_Y(x)\,(m).
\end{equation}
Indeed, if we have this equality, then by $\funext$ for cofibrant types (used at the second equality), we get
\[\alpha_Y(\beta_Y(x))=u_Y(\beta_Y(x)\ecirc r) = u_Y(v_Y(x))=x.\] We will prove the identity (\ref{aux-eq-nat}) by induction on $m:\NN^e$. For $m=\zero^e$, we get \[\beta_Y(x)(r(\zero^e))=\beta_Y(x)(\zero)=S(\zero,Y,x)=v_Y(x)(\zero).\] For $m=\suc^e m'$, we get by induction \[\beta_{Y'}(x')(r(m'))=v_{Y'}(x')(m').\] Using this, we obtain
\begin{eqnarray*}
\beta_Y(x)(r(m))&=&\beta_Y(x)(\suc (r (m')))\\
&=& S(\suc (r (m')), Y, x)\\
&=& S(r(m'), Y', x')\\
&=& \beta_{Y'}(x')(r(m')) \\
&=& v_{Y'}(x')(m')\\
&=& v_{Y'}(u_{Y'}(\lambda a.\,v_Y(x)\,(\suc^e (a))))\,(m')\\
&=& (\lambda a.\,v_Y(x)\,(\suc^e (a)))\,(m')=v_Y(x)\,(m) \qedhere
\end{eqnarray*}
\end{proof}

\textbf{Agda Side.} The folder \texttt{Sharpness} in the library \cite{agdalib} contains all the definitions and proofs in this section.

\section{Lifting cofibrancy from exo-nat to other types}\label{novel-sec}

In this section, we give a new result about other inductive types. Using cofibrant exo-nat, we will show that some other inductive types preserve cofibrancy.

\subsection{List exo-types}

\begin{defn}
For an exo-type $A:\U^e$ we define the exo-type $\List^e (A) : \U^e$ of \textbf{finite exo-lists} of terms of $A$, which has constructors 
\begin{itemize}
\item[•] $\nil^e : \List^e (A)$
\item[•] $\cons^e : A \rightarrow \List^e (A) \rightarrow \List^e (A)$
\end{itemize}

Similarly, if $A : \U$ is a type, the type $\List(A)$ of \textbf{finite lists} of $A$ has constructors $\nil$ and $\cons\,$.
\end{defn}

As in Theorem \ref{lemma211}, we have an obvious map $f : \List^e (A) \rightarrow \List (A)$ for a type $A$ defined as $f(\nil^e):=\nil$ and $f(a\cons^e l):=a \cons f(l)$. We will give some conditions for cofibrancy and sharpness of $\List^e (A)$. Indeed, if we assume $\NN^e$ and $A$ are cofibrant, then we can show $\List^e (A)$ is cofibrant. The proof is obtained by an isomorphism between a cofibrant exo-type and $\List^e(A)$. If $A$ is also sharp, the same isomorphism gives that $\List^e (A)$ is sharp. Moreover, we can show sharpness of $\List^e(A)$ in a way analogous to Proposition \ref{sharp-prop}(vi). 

\begin{lemma}\label{equivdef-list}
Let $A : \U^e$ be an exo-type. Then we have \[\left( {\sum_{n:\NN^e}}^e A^n \right) \cong \List^e(A)\] where $A^{\zero^e}:=\unit^e$ and $A^{(\suc^e n)} := A \times^e A^n$.
\end{lemma}

\begin{proof}
Define $\phi : \left(\sum^e_{n:\NN^e} A^n\right) \rightarrow \List^e(A)$ as follows:
\begin{eqnarray*}
\phi(\zero^e, \star^e)&:=& \nil^e\\
\phi(\suc^e(n), (a,p))&:=& a \cons^e \phi(n,p).
\end{eqnarray*}
Define $\theta : \List^e(A) \rightarrow \left(\sum^e_{n:\NN^e} A^n\right) $ as follows:
\begin{eqnarray*}
\theta(\nil^e)&:=& (\zero^e, \star^e)\\
\theta(a \cons^e l)&:=&(\,\suc^e(\pi^e_1\,\theta(l))\,,(a,\pi^e_2\,\theta(l))\,).
\end{eqnarray*}
Then it is easy to show by the induction on the constructors that \[\phi \ecirc \theta \exoeq \id_{\List^e(A)}\quad\text{and}\quad \theta \ecirc \phi \exoeq \id_{\sum^e_{n:\NN^e} A^n}.\qedhere\]
\end{proof}

\begin{prop}\label{cofib-list}
If $A:\U^e$ is cofibrant and $\NN^e$ is cofibrant, then so is $\List^e(A)$.
\end{prop}

\begin{proof}
By Lemma \ref{equivdef-list} and \ref{cofibrant-prop}(ii), it is enough to show that ${\sum^e_{n:\NN^e} A^n}$ is cofibrant. By the assumption and \ref{cofibrant-prop}(iii), we have that $A^n$ is a cofibrant exo-type for all $n:\NN^e$. Since we also assume $\NN^e$ is cofibrant, we are done by \ref{cofibrant-prop}(iv).
\end{proof}

By similar reasoning, we obtain that $\List^e(A)$ is sharp if $A$ is sharp and $\NN^e$ is cofibrant. However, we can also show the sharpness of $\List^e(A)$ like in Proposition \ref{sharp-prop}(vi).

\begin{prop}\label{sharp-list-2}
Let $A:\U^e$ be a sharp exo-type. Suppose also that $\NN^e$ is cofibrant. Then $\List^e(A)$ is sharp. 
\end{prop}

\begin{proof}
Since $A$ is sharp, it is cofibrant. By Proposition \ref{cofib-list}, we have $\List^e(A)$ is cofibrant, so it remains to find the fibrant replacement of it. 

Since $A$ is sharp, we have a type $RA:\U$ and a map $r_A:A\rightarrow RA$ such that the map \ref{precomp} is a fibrant-equivalence for any $Y:RA \rightarrow \U$. We claim that $\List (RA)$ is a fibrant replacement of $\List^e(A)$. Define $r:\List^e(A)\rightarrow \List (RA)$ as 
\begin{eqnarray*}
r(\nil^e)&:=&\nil\\
r(a \cons^e l)&:=&r_A(a) \cons r(l)
\end{eqnarray*}

Consider the following commutative diagram for any $Y: \List(RA) \rightarrow \U$:
\begin{center}
\begin{tikzcd}
\prod_{t:\List(RA)} Y(t) \arrow[dd, no head, shift right] \arrow[dd, no head] \arrow[rr, "(-\ecirc r)"] &  & \prod^e_{s:\List^e(A)} Y(r(s)) \arrow[dd, "u_Y"', shift right=2] \\
                                                                                                        &  & \cong                                                          \\
\prod_{t:\List(RA)} Y(t) \arrow[rr, "\alpha_Y:=u_Y \ecirc (-\ecirc r)"]                                     &  & FM_Y \arrow[uu, "v_Y"', shift right=2]                            
\end{tikzcd}
\end{center}

The type $FM_Y$ and the isomorphism $u_Y$ are obtained by the cofibrancy of $\List^e(A)$. We want to show that $\alpha_Y$ is an equivalence. First, we define an auxiliary type 
\[S : \prod_{t:\List(RA)} \left(\prod_{Y:\List(RA)\rightarrow \U}\left( \prod_{x:FM_Y} Y(t)\right)\right) \] by $S(\nil,Y,x):=v_Y(x)(\nil^e)$ and $S(c \cons l,Y,x):=S(l,Y',x')$ where
\begin{eqnarray*}
Y'(l)&:=&Y(c \cons l),\\
x'&:=& u_{Y'}(T)
\end{eqnarray*}
for a $T:\prod^e_{s:\List^e(A)}Y'(r(s))$ defined as follows: For $s:\List^e(A)$, consider the following diagram:
\begin{center}
\begin{tikzcd}
{\prod_{c:RA} Y(c \cons r(s))} \arrow[rr, "(-\ecirc r_A)"] \arrow[dd, no head, shift right] \arrow[dd, no head] &        & {\prod^e_{a:A} Y(r_A(a) \cons r(s))} \arrow[dd, "u"', shift right=2] \\
                                                                                                                       &        & \cong                                                                  \\
{\prod_{c:RA} Y(c \cons r(s))} \arrow[rr, "\alpha", shift left=2]                                               & \simeq & FM \arrow[ll, "\beta", shift left=2] \arrow[uu, "v"', shift right=2]  
\end{tikzcd}
\end{center}
The equivalence $\alpha$ is obtained by the sharpness of $A$. Now we define 
\[T(s):=\beta(u(\lambda\, a. v_Y(x)(a \cons^e s)))\, (c)\; : Y(c \cons r(s))=Y'(r(s)).\]
We also claim that for any $s:\List^e(A)$, $Y:\List(RA) \rightarrow \U$, and $x:FM_Y$ we have 
\begin{equation}\label{aux-eq-list}
S(r(s),Y,x)=v_Y(x) (s).
\end{equation}
It follows by induction on $\List^e(A)$. If $s=\nil^e$, the $\refl$ term satisfies the identity \ref{aux-eq-list}. If $s=b \cons^e s'$ for $s':\List^e(A)$, then we have the following chain of identities: 
\begin{eqnarray*}
S(r(b \cons^e s'),Y,x)&=&S(r_A(b) \cons r(s'),Y,x)\\
&=&S(r(s'),Y',x')\\
&=&v_{Y'}(u_{Y'} (T))(s')\\
&=&T(s')\\
&=&\beta(u(\lambda\, a. v_Y(x)(a \cons^e s')))\, (r_A(b))\\
&=&(\lambda\, a. v_Y(x)(a \cons^e s')) \,(b)\\
&=&v_Y(x)(b \cons^e s')
\end{eqnarray*}

These are obtained by, respectively, the definition of $r$, the definition of $S$, the induction hypothesis, the fact that $v_Y'$ is the inverse of $u_Y'$, the definition of $T$, the fact that $\beta$ is the inverse of $\alpha=u\ecirc(-\ecirc r)$, and the definition of the given function. Note that when we have exo-equalities of terms in types, we can use $\exotoid$ to make them identities.

Now, define $\beta_Y : FM_Y \rightarrow \prod_{t:\List(RA)} Y(s)$ as $\beta_Y(x)(t):=S(t,Y,x)$. Then we obtain \[\alpha_Y(\beta_Y(x))=u_Y(\beta_Y(x) \ecirc r)=u_Y(v_Y(x))=x. \] These are obtained by, respectively, the definition of $\alpha_Y$, the fact that $\beta_Y(x)\ecirc r= v_Y(x)$ since we can use $\funext^e$ for cofibrant exo-types and Equation \ref{aux-eq-list}, and the fact that $v_Y$ is the inverse of $u_Y$. 

This proves that $\alpha_Y$ has a section for any $Y:\List(RA) \rightarrow \U$. By Lemma \ref{equiv-defsharp}, we conclude that $\List^e(A)$ is sharp.
\end{proof}

\textbf{Agda Side.} The file \verb|Cofibrancy_of_List| in our Agda library \cite{agdalib} is the formalization of Lemma \ref{equivdef-list} and Proposition \ref{cofib-list}. The file \verb|On_Sharpness_of_List| is the formalization of Proposition \ref{sharp-list-2}.

\subsection{Exo-type of binary trees}

\begin{defn}
For exo-type $N,L:\U^e$ we define the exo-type $\BinTree^e (N,L) : \U^e$ of \textbf{binary exo-trees} with node values of exo-type $N$ and leaf values of exo-type $L$, which has constructors 
\begin{itemize}
\item[•] $\leaf^e : L \rightarrow \BinTree^e (N,L)$
\item[•] $\node^e : \BinTree^e (N,L) \rightarrow N \rightarrow \BinTree^e (N,L) \rightarrow \BinTree^e (N,L)$
\end{itemize}

Similarly, if $N,L : \U$ is a type, the type $\BinTree (N,L)$ of \textbf{binary trees} with node values of type $N$, leaf values of type $L$, and constructors $\leaf$ and $\node$.
\end{defn}

We also have a definition for unlabeled binary (exo)trees. 

\begin{defn}
The exo-type $\UnLBinTree^e : \U^e$ of \textbf{unlabeled binary exo-trees} is constructed by 
\begin{itemize}
\item[•] $\uleaf^e : \UnLBinTree^e $
\item[•] $\unode^e : \UnLBinTree^e \rightarrow \UnLBinTree^e \rightarrow \UnLBinTree^e $
\end{itemize}
Similarly, the type $\UnLBinTree : \U$ of \textbf{unlabeled binary trees} is constructed by $\uleaf$ and $\unode$.
\end{defn}

It is easy to see that if we take $N=L=\unit^e$, then $\BinTree^e (N,L)$ is isomorphic to $\UnLBinTree^e$. However, we have a more general relation between them. For any $N,L:\U^e$, we can show
\begin{equation}\label{tree-equiv-def}
\BinTree^e (N,L) \cong {\sum}^e_{t : \UnLBinTree^e} \left(N^{\text{\# of nodes of }t} \times^e L^{\text{\# of leaves of }t}\right).
\end{equation}

Thanks to this isomorphism, we can determine the cofibrancy or the sharpness of $\BinTree^e (N,L)$ using the cofibrancy or the sharpness of $\UnLBinTree^e$. Indeed, we will show that if $\NN^e$ is cofibrant, then $\UnLBinTree^e$ is not only cofibrant but also sharp. Since any finite product of cofibrant (sharp) exo-types is cofibrant (sharp), we can use isomorphism \ref{tree-equiv-def} to show $\BinTree^e (N,L)$ is cofibrant (sharp) under some conditions. Thus, the main goal is to get cofibrant (sharp) $\UnLBinTree^e$.

We will construct another type that is easily shown to be cofibrant (sharp) to achieve this goal. Let $\Parens : \U^e$ be the exo-type of parentheses constructed by $\popen : \Parens$ and $\pclose : \Parens$. In other words, it is an exo-type with two terms. Define an exo-type family \[\isbalanced : \List^e (\Parens) \rightarrow \NN^e \rightarrow \U^e \] where $\isbalanced(l,n):=\unit^e$ if the list of parentheses $l$ needs $n$ many opening parentheses to be a balanced parenthesization, and $\isbalanced(l,n):=\emptype^e$ otherwise. For example, we have
\begin{eqnarray*}
&&\isbalanced(\popen \cons^e \pclose \cons^e \nil^e,\zero^e)=\unit^e\\
&&\isbalanced(\popen \cons^e \pclose \cons^e \pclose \cons^e \nil^e,\suc^e(\zero^e))=\unit^e\\
&&\isbalanced(\popen \cons^e \pclose \cons^e \popen \cons^e \nil^e,\suc^e(\zero^e))=\emptype^e.
\end{eqnarray*}

In other words, the first says that ``()" is a balanced parenthesization, the second says that ``())" needs one opening parenthesis, and the third says that ``()(" is not balanced if we add one more opening parenthesis. 

Since $\emptype^e$ and $\unit^e$ are cofibrant (also sharp), we get $\isbalanced(l,n)$ is cofibrant (sharp) for any $l:\List^e(\Parens)$ and $n:\NN^e$. Finally, for any $n:\NN^e$ define 
\[\Balanced(n):={\sum}^e_{l:\List^e(\Parens)} \isbalanced(l,n).\]

\begin{lemma}\label{cofib/sharp-balanced}
If $\NN^e$ is cofibrant, then for any $n:\NN^e$, the exo-type $\Balanced(n)$ is both cofibrant and sharp.
\end{lemma}

\begin{proof}
By definition, $\Parens$ is a finite exo-type. Therefore, it is both cofibrant (Proposition \ref{cofibrant-prop}(iii)) and sharp (Proposition \ref{sharp-prop}(v)). Since we assume $\NN^e$ is cofibrant, Proposition \ref{cofib-list} shows that $\List^e(\Parens)$ is cofibrant, and Proposition \ref{sharp-list-2} shows that $\List^e(\Parens)$ is sharp.

Since $\isbalanced(l,n)$ is both cofibrant and sharp for any $l:\List^e(\Parens)$ and $n:\NN^e$, the exo-type $\Balanced(n)$ is cofibrant by Proposition \ref{cofibrant-prop}(iv) and sharp by Proposition \ref{sharp-prop}(iv).
\end{proof}

The exo-type that we use to show $\UnLBinTree^e$ is both cofibrant and sharp, is $\Balanced (\zero^e)$. The following result will be analogous to the combinatorial result that there is a one-to-one correspondence between full binary trees and balanced parenthesizations \cite{graphref}.

\begin{prop}\label{equiv-unbalanced}
There is an isomorphism $\UnLBinTree^e\cong \Balanced (\zero^e)$.
\end{prop}

\begin{proof}
We define the desired map by explain its construction. For the proof that it is indeed an isomorphism, we refer to its formalization in our library. 

Define first \[\phi : \UnLBinTree^e\rightarrow (n : \NN^e) \rightarrow \Balanced(n) \rightarrow \Balanced(n)\] as follows: 
\begin{eqnarray*}
\phi(\uleaf^e ,\,n,\,b)&:=&b\\
\phi(\unode^e(t_1,t_2),\,n,\,b)&:=&\phi(t_2,\,n,\,b')
\end{eqnarray*}
where 
\begin{eqnarray*}
b':=(\popen \cons^e (\pi_1(\phi\,(t_1,\,\suc^e(n),\,(\pclose \cons^e \pi_1(b),\pi_2(b)))))\;,\\
\pi_2(\phi\,(t_1,\,\suc^e(n),\,(\pclose \cons^e \pi_1(b),\pi_2(b))))).
\end{eqnarray*}
Using this, the main map $\Phi : \UnLBinTree^e \rightarrow \Balanced(\zero^e)$ is defined as \[\Phi(t):=\phi(t,\zero^e,(\nil^e,\star^e)).\] The construction basically maps each tree to a balanced parenthesization in the following way. A leaf is represented by the empty list of parentheses. If trees $t_1$ and $t_2$ have representations $l_1$ and $l_2$, then the tree $\unode^e(t_1,t_2)$ is represented by the list $l_2\,(\,l_1\,)$. Figure \ref{tree-conversion} provides some examples of this conversion. 

The inverse of $\Phi$ is defined precisely by reversing this process, but it needs some auxiliary definitions. One can see the formalization for the details. 
\end{proof}

\newpage

\begin{figure}[ht]
\begin{tikzcd}
        &                                                 & \bullet \arrow[ld, no head] \arrow[rd, no head] & {} \arrow[rr, dashed]                           &         & ()     \\
        & \bullet                                         &                                                 & \bullet                                         &         &        \\
        &                                                 & \bullet \arrow[ld, no head] \arrow[rd, no head] & {} \arrow[rr, dashed]                           &         & ()()   \\
        & \bullet                                         &                                                 & \bullet \arrow[ld, no head] \arrow[rd, no head] &         &        \\
        &                                                 & \bullet                                         &                                                 & \bullet &       \\
        &                                                 & \bullet \arrow[ld, no head] \arrow[rd, no head] & {} \arrow[rr, dashed]                           &         & (()()) \\
        & \bullet \arrow[ld, no head] \arrow[rd, no head] &                                                 & \bullet                                         &         &        \\
\bullet &                                                 & \bullet \arrow[ld, no head] \arrow[rd, no head] &                                                 &         &        \\
        & \bullet                                         &                                                 & \bullet                                         &         &        \\
\end{tikzcd}

\begin{tikzcd}
        &                                                 &                                                 & \bullet \arrow[lld, no head] \arrow[rrd, no head] &         & {} \arrow[rr, dashed]                           &         & ()(()()) \\
        & \bullet \arrow[ld, no head] \arrow[rd, no head] &                                                 &                                                   &         & \bullet \arrow[ld, no head] \arrow[rd, no head] &         &          \\
\bullet &                                                 & \bullet \arrow[ld, no head] \arrow[rd, no head] &                                                   & \bullet &                                                 & \bullet &          \\
        & \bullet                                         &                                                 & \bullet                                           &         &                                                 &         &         
\end{tikzcd}
\caption{Examples of the conversion between binary trees and parenthesization.}\label{tree-conversion}
\end{figure}
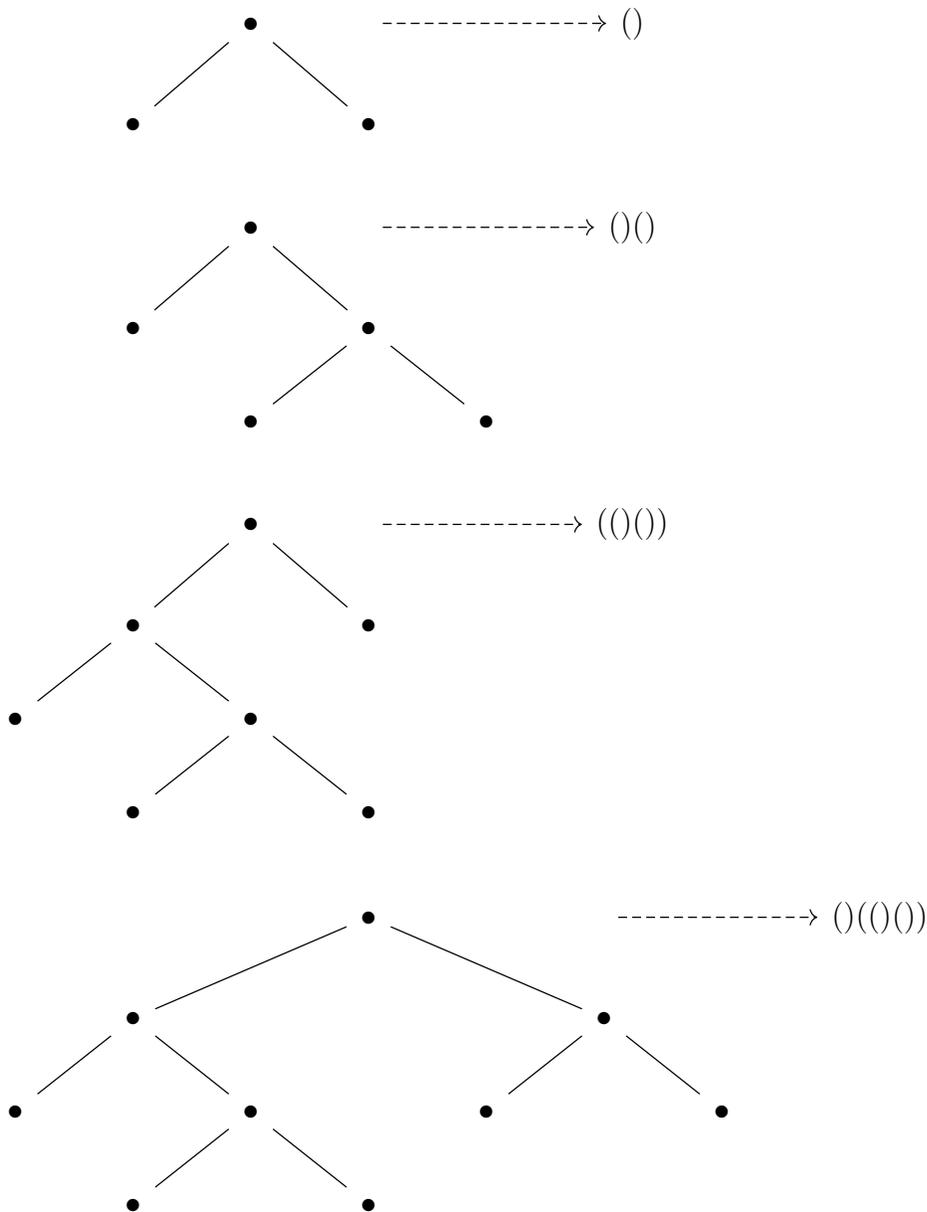 

This isomorphism provides the results we wanted.

\begin{cor}\label{unltree-cofib/sharp}
If $\NN^e$ is cofibrant, then $\UnLBinTree^e$ is both cofibrant and sharp.
\end{cor}

\begin{proof}
It follows from Lemma \ref{cofib/sharp-balanced} and Proposition \ref{equiv-unbalanced}.
\end{proof}

\begin{cor}
If $\NN^e$ is cofibrant, and $N,L:\U^e$ are cofibrant (sharp) exo-types, then $\BinTree^e(N,L)$ is cofibrant (sharp).
\end{cor}

\begin{proof}
It follows from the isomorphism \ref{tree-equiv-def} and Corollary \ref{unltree-cofib/sharp}.
\end{proof}

\textbf{Agda Side.} The file \texttt{BinTree} in our Agda library \cite{agdalib} include the proof of the isomorphism \ref{tree-equiv-def}. The file \verb|Cofibrancy_of_BinTree| includes the Proposition \ref{equiv-unbalanced} and the cofibrancy results about (unlabeled) binary trees. The file \verb|On_Sharpness_of_BinTree| includes the sharpness results about (unlabeled) binary trees.

\section{Semantics of two-level type theory}\label{semchap}

In this section, we will examine the semantic aspect of the theory discussed in the previous sections. In order for the axiom we accept about natural numbers to have meaning, we will investigate its models, preferably a large number of them. To do this, we will first provide the necessary background information about the models of the two-level type theory and then introduce the additional conditions required for the fulfillment of the aforementioned axiom.

During this section, we follow the conventions below.

\textbf{Variables.} $\Gamma$, $\Delta$, $\ldots$ stand for \emph{contexts}, $\sigma$, $\theta$, $\tau$, $\ldots$ for \emph{context morphisms}, $P$, $Q$, $R$, $\ldots$ for \emph{presheaves}, $A$, $B$, $C$, $Y$, $\ldots$ for \emph{types} and \emph{type families}, and $a$, $b$, $c$, $\ldots$ for \emph{terms}.

\textbf{Substitution.} Whenever $P : \C^{\op} \rightarrow \Set$ is a presheaf, $\sigma : \Delta \rightarrow \Gamma$ a morphism, and $A : P(\Gamma)$, we write $A[\sigma]$ instead of $P(\sigma)(A)$. 

\textbf{Equality signs.} Recall 2LTT has two different equality signs: ``$\exoeq$" and ``$=$". Now, another equality comes forward, that is the equality in \emph{metatheory}. We reserve ``$=$" for the metatheory's equality, and use ``$\Id$" for the identity type, ``$\Eq$" for the exo-equality.

\subsection{Category with families}

\begin{defn}\label{CwF-defn}
A category with families (CwF) consists of the following:
\begin{itemize} 
\item A category $\C$ with a terminal object $ 1_{\C} : \C$. Its objects are called \emph{contexts}, and $1_{\C}$ is called the \emph{empty context}.
\item A presheaf $\Ty : \C^{\op} \rightarrow \Set$. If $ A : \Ty(\Gamma)$, then we say $A$ \emph{is a type over} $\Gamma$.
\item A presheaf $\Tm : (\int \Ty)^{\op} \rightarrow \Set$. If $a : \Tm(\Gamma, A)$, then we say $a$ \emph{is a term of} $A$. 
\item For any $\Gamma : \C$ and $A : \Ty(\Gamma)$, there is an object $\Gamma.A :\C$, a morphism $p_A : \Gamma.A \rightarrow \Gamma$, and a term $q_A : \Tm(\Gamma.A,A[p_A])$ with the universal property: for any object $\Delta : \C$, a morphism $\sigma : \Delta \rightarrow \Gamma$, and a term $a:\Tm(\Delta,A[\sigma])$, there is a unique morphism $\theta : \Delta \rightarrow \Gamma.A$ such that $p_A \circ \theta = \sigma$ and $q_A[\theta]=a$. This operation is called the \emph{context extension}.
\end{itemize}
\end{defn}

Note that for all contexts $\Gamma : \C$ and types $A : \Ty(\Gamma)$, there is a natural isomorphism \[\Tm(\Gamma,A)\cong \C/{\Gamma}((\Gamma,\id_{\Gamma}),(\Gamma.A,p_A)).\] Indeed, this follows from the universal property of the context extension by taking $\Delta:=\Gamma$ and $\sigma:=\id_{\Gamma}$. This observation says that the terms of $A$ over $\Gamma$ can be regarded as the sections of $p_A:\Gamma.A \rightarrow \Gamma$.

The proposition below is a useful fact for the rest of the section.

\begin{prop}
Let $\sigma :\Delta \rightarrow \Gamma$ be a context morphism and $A:\Ty(\Gamma)$. There exists a morphism $\sigma^{+}:\Delta.A[\sigma] \rightarrow \Gamma.A$ that makes the following diagram into a pullback square:
\begin{center}
\begin{tikzcd}
{\Delta.A[\sigma]} \arrow[rr, "\sigma^{+}"] \arrow[dd, "{p_{A[\sigma]}}"'] &  & \Gamma.A \arrow[dd, "p_A"] \\
                                                                           &  &                            \\
\Delta \arrow[rr, "\sigma"']                                               &  & \Gamma                    
\end{tikzcd}.
\end{center}
\end{prop}

\begin{proof}
The existence of a morphism $\sigma^{+}$ follows from the universal property for the extension $\Gamma.A$, using the morphism $\sigma \circ p_{A[\sigma]}:\Delta.A[\sigma] \rightarrow \Gamma$. Consider another commutative diagram of the form
\begin{center}
\begin{tikzcd}
\Theta \arrow[rrrd, "\eta", bend left] \arrow[rddd, "\tau"', bend right] &                                                                            &  &                            \\
                                                                      & {\Delta.A[\sigma]} \arrow[rr, "\sigma^{+}"] \arrow[dd, "{p_{A[\sigma]}}"'] &  & \Gamma.A \arrow[dd, "p_A"] \\
                                                                      &                                                                            &  &                            \\
                                                                      & \Delta \arrow[rr, "\sigma"']                                               &  & \Gamma                    
\end{tikzcd}.
\end{center} Universal property for the extension $\Delta.A[\sigma]$ gives a unique morphism $\theta :\Theta \rightarrow \Delta.A[\sigma]$ such that $p_{A[\sigma]} \circ \theta = \tau$. Since we have both \[p_A\circ \eta=\sigma \circ \tau\] and \[p_A\circ \sigma^{+} \circ \theta= \sigma \circ p_{A[\sigma]} \circ \theta = \sigma \circ \tau,\] by the universal property for the extension $\Gamma.A$, we have $ \sigma^{+} \circ \theta = \eta$.
\end{proof}

Rather than presenting a specific instance of a CwF, we will offer a more extensive range of CwF examples in the subsequent section.

\subsubsection*{Presheaf CwFs}\label{presheaf-CwF}

The category of presheaves is an archetypal example of a CwF \cite{hofmann}. Let $\C$ be a (small) category, and $\widehat{\C}$ be its category of presheaves. The CwF structure on $\widehat{\C}$, denoted by $(\widehat{\Ty}, \widehat{\Tm})$, is defined in the following manner:

\begin{itemize}
\item Contexts are presheaves $\C^{\op} \rightarrow \Set$.
\item The constant presheaf that takes the value ${\star}$ in the category of sets can be characterized as the terminal object $1_{\widehat{\C}}$.
\item Recall that $\widehat{\Ty}$ is a presheaf on $\widehat{\C}$. 

If $P:\widehat{\C}$, then $\widehat{\Ty}(P)$ is the underlying set of the category of presheaves $\widehat{\int P}$ over \emph{the category of elements} $\int P$. In other words, a type $A:\widehat{\Ty}(P)$ is a functor $\left(\int P \right)^{\op} \rightarrow \Set$. 

If $\phi:Q \rightarrow P$ is a morphism in $\widehat{\C}$ and $A:\widehat{\Ty}(P)$, we define the type substitution $A[\phi]:\widehat{\Ty}(Q)$ as \[A[\phi](\Gamma,x):= A (\Gamma,\phi_{\Gamma}(x))\] where $x:Q_{\Gamma}$.

\item Recall that $\widehat{\Tm}$ is a presheaf on $\int \widehat{\Ty}$. 

For $P:\widehat{\C}$ and $A:\widehat{\Ty}(P)$, we define \[\widehat{\Tm}(P,A) := \left\lbrace a : \prod_{\Gamma : \C,\, x: P_{\Gamma}} A(\Gamma,x) \mid \text{ if } \sigma : \Delta \rightarrow \Gamma, \, x : P_{\Gamma}\text{, then } a(\Gamma,x)[\sigma]=a(\Delta,x[\sigma])\right\rbrace.\] 

If $\phi:Q \rightarrow P$ is a morphism in $\widehat{\C}$, $A:\widehat{\Ty}(P)$, and $a:\widehat{\Tm}(P,A)$, we define the term substitution $a[\phi]:\widehat{\Tm}(Q,A[\phi])$ as \[a[\phi](\Gamma,x):= a (\Gamma,\phi_{\Gamma}(x))\] where $x:Q_{\Gamma}$.

\item For $P:\widehat{\C}$ and $A:\widehat{\Ty}(P)$, the context $P.A$ is again a presheaf over $\C$ defined by \[P.A(\Gamma):=\coprod_{x : P(\Gamma)} A(\Gamma, x).\] 

If $\sigma : \Delta \rightarrow \Gamma$ is a morphism in $\C$ and $(x,a):P.A (\Gamma)$, then we define \[(x,a)[\sigma]=P.A(\sigma) (x,a):=(x[\sigma],a[\sigma]).\]

The morphism $p_A:P.A \rightarrow P$ is defined by the first projection. In other words, for $\Gamma : \C$ and $(x ,a): P.A(\Gamma)$, we have $(p_A)_{\Gamma}(x,a)=x$. The term $q_A: \widehat{\Tm}(P.A,A[p_A])$ is given by the second projection. In other words, for $\Gamma:\C$ and $(x,a): P.A(\Gamma)$, we have $q_A(\Gamma, (x,a))=a$. Note that $A[p_A](\Gamma,(x,a))=A(\Gamma,p_A(x,a))=A(\Gamma,x)$.

It remains to verify the universal property for the context extension. 

Let $Q:\widehat{\C}$, $\tau:Q\rightarrow P$, and $b:\widehat{\Tm}(Q,A[\tau])$. Define $\theta :Q \rightarrow P.A$ as follows: for $\Gamma : \C$ and $x:Q_{\Gamma}$, we have \[\theta_{\Gamma}(x):= (\tau_{\Gamma}(x),b(\Gamma,x)).\] It is straightforward to verify the defining rules: \underline{$p_A\circ \theta =\tau$} because \[(p_A\circ \theta)_{\Gamma}(x)=p_A(\tau_{\Gamma}(x),b(\Gamma,x))=\tau_{\Gamma}(x),\] and \underline{$q_A[\theta]=b$} because \[q_A[\theta](\Gamma,x)=q_A(\Gamma, \theta_{\Gamma}(x))=q_A(\Gamma, (\tau_{\Gamma}(x),b(\Gamma,x)))=b(\Gamma,x).\] Since the defining properties of $p_A$ and $q_A$ determines the map $\theta$, it is uniquely determined. 
\end{itemize}

Therefore $(\widehat{\C},\widehat{\Ty},\widehat{\Tm})$ satisfies the conditions in Definition \ref{CwF-defn}.

\subsubsection*{Simplicial set CwF}

Let $\C=\bigtriangleup$ be the simplex category. The presheaf category $\widehat{\bigtriangleup}$ is called \emph{the category of simplicial sets}, denoted by $\SSet$. Like any other presheaf category, $\SSet$ has a CwF structure $(\widehat{\Ty},\widehat{\Tm})$, but it has another CwF structure. We only define a new type presheaf $\Ty : \SSet^{\op} \rightarrow \Set$ as follows: $\Ty(P)$ is a subset of $\widehat{\Ty}(P)$ such that $A : \Ty(P)$ if the display map $P.A \rightarrow P$ is a \emph{Kan fibration}. With the induced term presheaf $\Tm$ obtained by $\widehat{\Tm}$, we obtain a CwF that is $(\SSet, \Ty, \Tm)$. It is easy to prove that this structure satisfies all axioms of being a CwF \cite{paolo-thesis}. We always refer to this new CwF structure on $\SSet$ unless stated otherwise. 

$\SSet$ is not a unique example of a presheaf category having two CwF structures. There are many of them, and it will be helpful to talk about two-level structures in the later sections.

\subsection{Type formers in CwFs}

The objective of this section is to establish the meanings of particular type formers within a CwF and to examine the requirements that must be fulfilled for the CwF to possess these type formers. Although this analysis could be applied to various standard type formers, our focus will be solely on those indispensable for the subsequent sections.

\subsubsection*{Dependent function types} 

We say that a CwF $(\C, \Ty, \Tm)$ \textit{supports $\Pi$-types} \cite{hofmann} if 
\begin{itemize}
\item for any two types $A:\Ty(\Gamma)$ and $B:\Ty(\Gamma.A)$, there is a type $\Pi(A,B):\Ty(\Gamma)$,
\item for each $b:\Tm(\Gamma.A,B)$, there is a term $\lambda(b):\Tm(\Gamma,\Pi(A,B))$, and
\item for each $f:\Tm(\Gamma,\Pi(A,B))$ and $a:\Tm(\Gamma,A)$, there is a term $\app(f,a):\Tm(\Gamma,B[a])$
\end{itemize}
such that the following equations (with appropriate quantifiers) hold:
\begin{eqnarray*}
&&\app(\lambda(b),a)=b[a]\\
&&\lambda(\app(f,a),q_A)=f\\
&&\Pi(A,B)[\tau]=\Pi(A[\tau],B[\tau^+])\\
&&\lambda(b)[\tau]=\lambda(b[\tau])\\
&&\app(f,a)[\tau]=\app(f[\tau],a[\tau]).
\end{eqnarray*}

Note that using dependent function types, one can define simple function types. Indeed, if $A,B:\Ty(\Gamma)$, then \textit{the type of functions} from $A$ to $B$ over $\Gamma$, denoted by $B^A:\Ty(\Gamma)$ is defined via $\Pi(A,B[p_A])$.

We will show that a presheaf CwF supports $\Pi$-types. 

\begin{prop}
For any (small) category $\C$, the presheaf CwF $(\widehat{\C},\widehat{\Ty},\widehat{\Tm})$ supports $\Pi$-types.
\end{prop}

\begin{proof}
Let $A:\widehat{\Ty}(P)$ and $B:\widehat{\Ty}(P.A)$. First, we need to define $\Pi(A,B):\widehat{\Ty}(P)$. Recall that $\Pi(A,B)$ should be a presheaf over $\int P$. For each $\Gamma:\C$ and $x:P_{\Gamma}$, the type $\Pi(A,B)(\Gamma,x)$ consists of the elements $f$ in the categorical product
\begin{equation*}
f: \mathbf{\prod}_{\Delta:\C,\, \sigma:\Delta\rightarrow \Gamma,\, a: A (\Delta, x[\sigma])} B(\Delta, (x[\sigma],a))
\end{equation*} 
such that if $\Theta : \C$ and $\tau : \Theta \rightarrow \Delta$, then 
\begin{equation}\label{comp-cond}
f(\Delta,\sigma,a)[\tau]=f(\Theta,\sigma \circ \tau,a[\tau]). 
\end{equation}
If $\tau:(\Upsilon,y)\rightarrow(\Gamma,x)$ is a morphism in $\int P$, namely, $\tau : \Upsilon \rightarrow \Gamma$ is a morphism in $\C$ such that $x[\tau]=y$, for each $f : \Pi(A,B)(\Gamma,x)$, $\Delta : \C$, $\sigma : \Delta \rightarrow \Upsilon$, and $a : A (\Delta, y[\sigma])$, we define $f[\tau](\Delta,\sigma,a):=f(\Delta, \tau \circ \sigma, a)$. Using the compatibility condition \ref{comp-cond}, we indeed obtain a presheaf $\Pi(A,B)$ on $\int P$.

Second, for each $b:\widehat{\Tm}(P.A,B)$, we need to define a term $\lambda(b):\widehat{\Tm}(P,\Pi(A,B))$. Recall that $\lambda(b)$ should be an element in \[\mathbf{\prod}_{\Gamma : \C, x : P_{\Gamma}} \Pi(A,B)(\Gamma,x).\] Now, for $\Gamma:\C$, $x : P_{\Gamma}$, $\Delta:\C$, $\sigma:\Delta \rightarrow \Gamma$, and $a : A(\Delta, x[\sigma])$ we define \[\lambda(b)(\Gamma,x) (\Delta,\sigma,a):=b(\Delta,x[\sigma],a).\] This definition makes sense because the term $b$ is in $\mathbf{\prod}_{\Gamma : \C, z : P.A_{\Gamma}} B(\Gamma,z)$.

Third, for each $f:\widehat{\Tm}(P,\Pi(A,B))$ and $a:\widehat{\Tm}(P,A)$, we need to define a term $\app(f,a):\widehat{\Tm}(P,B[a])$. Recall that it should be in \[\mathbf{\prod}_{\Gamma : \C, x : P_{\Gamma}} B[a](\Gamma,x).\] Now, for $\Gamma:\C$, $x : P_{\Gamma}$, we define \[\app(f,a)(\Gamma,x):= f(\Gamma,x)(\Gamma,\id,a).\]
It is easy but straightforward to prove coherence rules \cite{hofmann}.
\end{proof}

\subsubsection*{Dependent pair types}

We say that a CwF $(\C, \Ty, \Tm)$ \textit{supports $\Sigma$-types} \cite{hofmann} if 
\begin{itemize}
\item for any two types $A:\Ty(\Gamma)$ and $B:\Ty(\Gamma.A)$, there is a type $\Sigma(A,B):\Ty(\Gamma)$,
\item for each $a:\Tm(\Gamma,A)$ and $b:\Tm(\Gamma,B[a])$, there is a term $\langle a,b \rangle :\Tm(\Gamma,\Sigma(A,B))$, and
\item for each $z: \Tm(\Gamma,\Sigma(A,B))$, there are terms $\pi_1(z):\Tm(\Gamma,A)$ and $\pi_2(z):\Tm(\Gamma,B[\pi_1(z)])$
\end{itemize}
such that the following equations (with appropriate quantifiers) hold:
\begin{eqnarray*}
&&\pi_1 (\langle a,b \rangle)=a \\
&&\pi_2 (\langle a,b \rangle)=b \\
&&\langle \pi_1(z),\pi_2(z) \rangle = z \\
&&\Sigma(A,B)[\tau]=\Sigma(A[\tau],B[\tau^+])\\
&&\langle a,b \rangle [\tau] = \langle a[\tau],b[\tau] \rangle \\
&&\pi_1(z)[\tau]=\pi_1(z[\tau])\\
&&\pi_2(z)[\tau]=\pi_2(z[\tau]).
\end{eqnarray*}

We will show that a presheaf CwF supports $\Sigma$-types. 

\begin{prop}
For any (small) category $\C$, the presheaf CwF $(\widehat{\C},\widehat{\Ty},\widehat{\Tm})$ supports $\Sigma$-types.
\end{prop}

\begin{proof}
Let $A:\widehat{\Ty}(P)$ and $B:\widehat{\Ty}(P.A)$. First, we need to define $\Sigma(A,B):\widehat{\Ty}(P)$. Recall that $\Sigma(A,B)$ should be a presheaf over $\int P$. For each $\Gamma:\C$ and $x:P_{\Gamma}$, we define \[\Sigma(A,B)(\Gamma,x):=\lbrace (a,b) |\, a : A(\Gamma,x),\, b: B(\Gamma,(x,a)) \rbrace. \] For a morphism $\sigma : (\Delta,y)\rightarrow(\Gamma,x)$ in $\int P$ and $(a,b):\Sigma(A,B)(\Gamma,x)$, we define \[(a,b)[\sigma]:=(a[\sigma],b[\sigma]).\] One can define the operations $\langle \_,\_ \rangle$, $\pi_1$, and $\pi_2$ in an obvious way, and it is easy to prove the coherence rules.
\end{proof}

For presheaf CwFs, there is a relation between $\Pi(A,B)$ and $\Sigma(A,B)$.

\begin{prop}
In the presheaf CwF $(\widehat{\C},\widehat{\Ty},\widehat{\Tm})$, if $A:\widehat{\Ty}(P)$ and $B:\widehat{\Ty}(P.A)$, then the type $\Pi(A,B)$ is a pullback for the diagram
\begin{equation*}
\begin{tikzcd}
                    &  & {{\Sigma (A,B)}^{A}} \arrow[dd, "\phi"] \\
                    &  &                                                  \\
1 \arrow[rr, "\id"'] &  & {A}^{A}                                 
\end{tikzcd}
\end{equation*}
where $1$ is the constant presheaf over $\int P$ and $\phi$ is given by the first projection.
\end{prop}

\begin{proof}
Note that all objects in the diagram are preheaves over $\int P$.

Define a natural transformation $\psi: \Pi(A,B) \rightarrow {\Sigma (A,B)}^{A}$ as follows: for each $\Gamma : \C$ and $x : P_{\Gamma}$, we have $\psi_{\Gamma,x}(f):=\tilde{f}$ where $\tilde{f}(a):=(a,f(a))$. It is easy to see that $\phi \circ \psi = \id$. Consider another commutative diagram:
\begin{equation}\label{another-diagram}
\begin{tikzcd}
D \arrow[rr, "u"] \arrow[dd] &  & {{\Sigma (A,B)}^{A}} \arrow[dd, "\phi"] \\
                             &  &                                                  \\
1 \arrow[rr, "\id"']          &  & {A}^{A}                                 
\end{tikzcd}.
\end{equation}
Define a natural transformation $\tau:D\rightarrow \Pi(A,B)$ as follows: for each $\Gamma : \C$ and $x : P_{\Gamma}$, we have $\tau_{\Gamma,x}(d)(a):=\pi_2(u_{\Gamma,x}(d)(a))$. It is easy to see that $\psi \circ \tau = u$, and $\tau$ is unique. This finishes our claim.
\end{proof}

\subsubsection*{Extensional identity type}

We say that a CwF $(\C, \Ty, \Tm)$ \textit{supports extensional identity types} \cite{hofmann} if
\begin{itemize}
\item for any type $A : \Ty(\Gamma)$, there is a type $\Eq_A : \Ty(\Gamma.A.A[p_A])$,
\item a morphism $\refl^e_A : \Gamma.A \rightarrow \Gamma.A.A[p_A].\Eq_A$ such that $p_{\Eq_A}\circ \refl^e_A $ equals the diagonal morphism $\Gamma.A \rightarrow \Gamma.A.A[p_A]$, and
\item for each $B : \Ty(\Gamma.A.A[p_A].\Eq_A)$, a function \[J^e:\Tm(\Gamma.A, B[\refl^e_A]) \rightarrow \Tm(\Gamma.A.A[p_A].\Eq_A, B)\] 
\end{itemize}
such that these data are stable under substitution with respect to context morphisms and such that
\begin{itemize}
\item if $h:\Tm(\Gamma.A, B[\refl^e_A])$, then $J^e(h)[\refl^e_A]=h$, and
\item if $h:\Tm(\Gamma.A.A[p_A].\Eq_A, B)$, then $J^e(h[\refl^e_A])=h$.
\end{itemize} 

The last equality can be thought of as an $\eta$ rule. Moreover, this rule holds if and only if $\refl^e_A$ is an isomorphism with the inverse $p_{A[p_A]} \circ p_{\Eq_A}$.

We will show that a presheaf CwF supports extensional identity types.

\begin{prop}
For any (small) category $\C$, the presheaf CwF $(\widehat{\C},\widehat{\Ty},\widehat{\Tm})$ supports extensional identity types.
\end{prop}

\begin{proof}
Let $A:\widehat{\Ty}(P)$. Recall that $\Eq_A$ should be a presheaf over $\int \Gamma.A.A[p_A]$. For each $\Gamma:\C$, $x:P_{\Gamma}$, and $a,b: A(\Gamma,x)$, we define 
\[\Eq_A(\Gamma,x,a,b):= 
    \begin{cases}
        \{\star\} & \text{if } a = b\\
        \emptyset & \text{if } a \neq b
    \end{cases}.\] 
The morphism $\refl^e_A : P.A \rightarrow P.A.A[p_A].\Eq_A$ is defined for each $\Gamma : \C$ as $(x,a) \mapsto (x,a,a,\star)$. Since $p_{\Eq}: P.A.A[p_A].\Eq_A \rightarrow P.A.A[p_A]$ is the first projection, we indeed have $p_{\Eq}\circ \refl^e_A$ is the diagonal map.

If $B : \widehat{\Ty}(P.A.A[p_A].\Eq_A)$, the function $J^e:\widehat{\Tm}(P.A, B[\refl^e_A]) \rightarrow \widehat{\Tm}(P.A.A[p_A].\Eq_A, B)$ is defined as follows: If $\alpha: \widehat{\Tm}(P.A, B[\refl^e_A])$, this means we have
\[\alpha : \prod_{\substack{\Gamma : \C \\ (x,a): P.A_{\Gamma}}} B(\Gamma,(x,a,a,\star))\quad \text{ and } \quad J^e(\alpha) : \prod_{\substack{\Gamma : \C\\ (x,a,b,p): P.A.A[p_A].\Eq_{\Gamma}}} B(\Gamma,(x,a,b,p)).\] Thus, we define $J^e(\alpha)(\Gamma,(x,a,b,p)):=\alpha(\Gamma,(x,a))$ which makes sense because $p : \Eq(a,b)$ means $a=b$ and $p=\star$. Now, it is easy to prove the coherence rules.
\end{proof}

\subsubsection*{Intensional identity type}

We say that a CwF $(\C, \Ty, \Tm)$ \textit{supports intensional identity types} \cite{hofmann} if
\begin{itemize}
\item for any type $A : \Ty(\Gamma)$, there is a type $\Id_A : \Ty(\Gamma.A.A[p_A])$,
\item a morphism $\refl_A : \Gamma.A \rightarrow \Gamma.A.A[p_A].\Id_A$ such that $p_{\Id_A}\circ \refl_A $ equals the diagonal morphism $\Gamma.A \rightarrow \Gamma.A.A[p_A]$, and 
\item for each $B : \Ty(\Gamma.A.A[p_A].\Id_A)$, a function \[J:\Tm(\Gamma.A, B[\refl_A]) \rightarrow \Tm(\Gamma.A.A[p_A].\Id_A, B)\] 
\end{itemize}
such that these data are stable under substitution with respect to context morphisms, and such that if $h:\Tm(\Gamma.A, B[\refl_A])$, then $J(h)[\refl_A]=h$.

The last equality can be thought of as a $\beta$ rule.

Since $\Id$ is a particular case of $\Eq$, we can say that every presheaf CwF supports intensional identity types. If we also assume \emph{Univalence} for intensional identities, it is not true in general that every presheaf CwF supports such an identity. For example, the presheaf CwF on $\Set$ supports the identity type and the uniqueness of identity proof (UIP), but UIP contradicts with univalence \cite{hott}. Nonetheless, it has been established that the simplicial set CwF, denoted as $\SSet$, does provide support for univalent identity types \cite{kap-lum}. Indeed, this category stands as one of the widely recognized models not only for Martin-Löf Type Theory but also for Homotopy Type Theory.

In the case of intensional identity types, we can also establish a definition for what constitutes a ``contractible" type. This notion serves as a crucial component in the overall definition of cofibrancy.

\begin{defn}
Let $A:\Ty(\Gamma)$ be a type. We call $A$ a \emph{contractible type} if there is a term in the following type over $\Gamma$: \[\iscontr(A):=\Sigma(A,\Pi(A[p_A],\Id_A)).\]
\end{defn}

Having such a term means that there is a term $c : \Tm(\Gamma,A)$ called \emph{center of contraction}, such that for any term $a: \Tm(\Gamma.A,A[p_A])$ there is a term $p: \Tm(\Gamma,\Id_A[c^+,a^+])$. In other words, we have a section map $c:\Gamma \rightarrow \Gamma.A$ to $p_A$ such that the following diagram, where $h$ is the contracting homotopy, commutes:
\begin{center}
\begin{tikzcd}
\Gamma.A \arrow[rd, "c^+"'] \arrow[rr, "h"] &                   & {\Gamma.A.A[p_A].\Id_A} \arrow[ld, "p_{\Id_A}"] \\
                                                  & {\Gamma.A.A[p_A]} &                                                
\end{tikzcd}.
\end{center}

\subsubsection*{Natural number type}

We say that a CwF $(\C, \Ty, \Tm)$ \textit{supports a natural number type} if
\begin{itemize}
\item there is a type $\NN : \Ty(1_{\C})$ where $1_{\C}$ is the terminal object of $\C$,
\item there is a term $\zero:\Tm(1_{\C}, \NN)$ which can be thought as a context morphism $1_{\C} \rightarrow 1_{\C}.\NN$,
\item there is a morphism $\suc:\Tm(1_{\C},\NN) \rightarrow \Tm(1_{\C},\NN)$ which can be thought as a context morphism $1_{\C}.\NN \rightarrow 1_{\C}.\NN$, and
\item for each $\Gamma : \C$, the unique morphism $\sigma : \Gamma \rightarrow 1_{\C}$, and $B : \Ty(\Gamma.\NN[\sigma])$ with two context morphisms $b_0:\Gamma \rightarrow \Gamma.\NN[\sigma].B$ and $b_s:\Gamma.\NN \rightarrow \Gamma.\NN[\sigma].B \rightarrow \Gamma.\NN[\sigma].B$, there is a morphism $J^{\NN}_B : \Gamma.\NN[\sigma] \rightarrow \Gamma.\NN[\sigma].B$
\end{itemize}
such that $J^{\NN}_B \circ \zero[\sigma]=b_0$ and $J^{\NN}_B \circ \suc[\sigma] (\alpha) = b_s (\alpha, \,J^{\NN}_B(\alpha))$, and these data are stable under substitution with respect to context morphisms.

The last morphism can be thought of as a usual induction rule on $\NN$.

\begin{prop}
For any (small) category $\C$, the presheaf CwF $(\widehat{\C},\widehat{\Ty},\widehat{\Tm})$ supports a natural number type.
\end{prop}

\begin{proof}
Since $\NN$ should be a presheaf over $\int 1_{\widehat{C}}$, for each $\Gamma:\C$ and $x:(1_{\widehat{C}})_{\Gamma}$ we define $\NN(\Gamma,x):= \mathbf{N}$, the external set of natural numbers. The term $\zero : \widehat{\Tm}(1_{\widehat{C}}, \NN)$ is obtained by the morphism $1_{\widehat{C}} \rightarrow 1_{\widehat{C}}.\NN$ which we define $\zero_{\Gamma}(x):=(x,0)$ for any $\Gamma:\C$ and $x:(1_{\widehat{C}})_{\Gamma}$. The morphism $\suc:1_{\widehat{C}}.\NN \rightarrow 1_{\widehat{C}}.\NN$ is defined as $\suc_{\Gamma}(x,k):=(x, k+1)$ for any $\Gamma:\C$ and $(x,k):(1_{\widehat{C}}.\NN)_{\Gamma}$. 

For any $P :\widehat{\C}$ and $\sigma : P \rightarrow 1_{\widehat{\C}}$, if $B : \widehat{\Ty}(P.\NN[\sigma])$ with $b_0$ and $b_s$, the function $J^{\NN}_B:P.\NN[\sigma] \rightarrow P.\NN[\sigma].B$ is defined as, for each $\Gamma:\C$, and $(x,k):(P.\NN[\sigma])_{\Gamma}$ 
\[(J^{\NN}_B)_{\Gamma}(x,k):= 
    \begin{cases}
        {b_0(x)} & \text{if } k = 0\\
        {b_s\big((x,k'),\,(J^{\NN}_B)_{\Gamma}(x,k')\big)} & \text{if } k = k' + 1
    \end{cases}.\] 

Clearly, $J^{\NN}_B \circ \zero[\sigma]=b_0$ holds. For the other equality, we have
\[(J^{\NN}_B \circ \suc[\sigma])_{\Gamma}(x,k)=(J^{\NN}_B)_{\Gamma}(x,k+1)=(b_s)_{\Gamma}\big((x,k),\,(J^{\NN}_B)_{\Gamma}(x,k)\big).\]
Also, it is easy to prove the remaining coherence rules.
\end{proof}

We can talk about the unit type, the empty type, coproducts, and the others and give the conditions for a CwF to support them. In general, if we have a collection $T$ of type formers, we say a \emph{CwF supports $T$} if the CwF supports each type formers in $T$. 

\begin{example}
Let $T:=\{\prod, \sum, \unit, \emptype, \NN, \Eq\}$. Then for any (small) category $\C$, the presheaf CwF $(\widehat{\C},\widehat{\Ty},\widehat{\Tm})$ supports $T$.\footnote{We've not given the proofs for $\unit$ and $\emptype$, but these are trivial facts.}
\end{example}

\subsubsection*{CwFs as a model}

Based on the example of presheaf CwF we discussed earlier, it becomes clear that CwFs can be used to model dependent type theory. A CwF provides a structure that covers contexts, types, terms, and substitutions. Moreover, if a CwF has enough type formers, it allows us to work with dependent products, dependent sums, and other inductive types. A reader who is interested in delving into the details and exploring similar constructions related to CwFs can refer to \cite{paolo-thesis}. Our primary focus here is to establish the necessary background to discuss models where exo-nat is cofibrant.

\begin{defn}[\cite{2ltt}]
A \emph{model of Martin-Löf type theory with type formers $T$} is a CwF that supports $T$.
\end{defn}

We already know that a presheaf CwF is a model of Martin-Löf type theory with usual type formers and extensional identity types. The simplicial set CwF $\SSet$ is a model of Martin-Löf type theory with usual type formers and intensional identity types.

\subsection{Two-level CwFs}

Let us recall that the aim of this study on semantics is to gain insight into the models of 2LTT with a cofibrant exo-nat. Once we have acquired models of various type theories, a natural question arises: can we merge these models to obtain a comprehensive model of 2LTT? The answer to this question is affirmative, as it is indeed possible to combine two CwF structures in the same category in a manner that ensures their compatibility and coherence.

\begin{defn}[\cite{paolo-thesis}]\label{two-level-CwF}
A \emph{two-level CwF} is a CwF $(\C,\Ty,\Tm)$, equipped with another type presheaf $\Ty^f :\C^{\op} \rightarrow \Set$, and a natural transformation $\coer : \Ty^f \rightarrow \Ty$.
\end{defn} 

\begin{remark}
Given a two-level CwF $\C$, we define a second CwF structure on $\C$ using $\Ty^f$ as the type functor and the term functor is obtained as $\Tm^f(\Gamma,A):=\Tm(\Gamma,\coer_{\Gamma}(A))$. The context extension holds from $\Gamma.A:=\Gamma.\coer_{\Gamma}(A)$. In order to emphasize the difference, we use the superscripts $\_^e,\_^f$ and write $\Ty^e,\Tm^e$ for the original CwF structure, write $\Ty^f,\Tm^f$ for the one obtained by the coercion transformation. It is not surprising that this choice is intentional to be consistent with the first section. The original CwF will model the ``exo" level of 2LTT, while the other model is the usual ``HoTT" level of 2LTT.
\end{remark}

\begin{example}
Recall $\widehat{\Ty}$ denotes the presheaf CwF structure. The simplicial set presheaf $\SSet$ originally have already a presheaf CwF structure. Recall that 
\[\Ty(P):=\{A : \widehat{\Ty}(P) \mid p_A : P.A \rightarrow P \text{ is a Kan fibration}\}\] gives another type functor. Taking $\Ty^f=\Ty$ and $\coer : \Ty^f \rightarrow \widehat{\Ty}$ as the inclusion, we obtain $\SSet$ as a two-level CwF.
\end{example}

In a similar vein to how we can build a presheaf CwF from any arbitrary (small) category when the category itself is a CwF, we can proceed to construct a two-level CwF. This particular construction, which we refer to as the \emph{presheaf two-level CwF}, will serve as our primary focus and model of interest.

\begin{defn}\label{presheaf-twolevel}
Let $\C$ be a (small) category with CwF structure $\Ty$, $\Tm$. There is a two-level CwF structure on $\widehat{\C}$ called \emph{presheaf two-level CwF}, denoted by $(\widehat{\C},\widehat{\Ty}, \widehat{\Tm},\Ty^f,\Tm^f)$ defined as follows:
\begin{itemize}
\item $(\widehat{\C},\widehat{\Ty}, \widehat{\Tm})$ is the presheaf CwF defined in Section \ref{presheaf-CwF}, 
\item given $P$ in $\widehat{\C}$, the type functor $\Ty^f$ is given by $\Ty^f(P):=\widehat{\C}(P,\Ty)$, and
\item for $\Gamma$ in $\C$ and $B$ in $P(\Gamma)$, we define $\coer_P(A)(\Gamma,B):=\Tm(\Gamma,A_{\Gamma}(B))$.
\end{itemize}
As before, given $A$ in $\Ty^f(P)$, we define $\Tm^f(P,A):=\widehat{\Tm}(P,\coer_P(A))$. 
\end{defn}

\subsubsection*{Two-level CwFs as a model}

Similar to how a type theory can be interpreted within the framework of a category with families, two-level type theories can be interpreted using a two-level category with families. Below, we provide the precise definition for such an interpretation.

\begin{defn}\label{two-level-model}
A \emph{two-level model} of a type theory with \emph{type formers $T^f$} and \emph{exo-type formers $T^e$} is a two-level CwF on a category $\C$ such that
\begin{itemize}
\item the structure $\Ty^f,\Tm^f$ is a model of type theory with $T^f$,
\item the structure $\Ty^e,\Tm^e$ is a model of type theory with $T^e$.
\end{itemize}
\end{defn}

\begin{remark}
In the subsequent sections, when we say \emph{a two-level model with enough type formers} or \emph{a model of 2LTT}, we mean the two-level model with $T^f$ and $T^e$ where the collections are the types and exo-types we defined in Section \ref{types&exo-types}. 
\end{remark}

Our assumption concerning the coercion morphism $\Ty^f \rightarrow \Ty^e$ is only that it is a natural transformation. However, in the context of a two-level model with enough type formers, we have the semantic counterpart of Theorem \ref{lemma211}. This theorem furnishes excellent inversion rules from types to exo-types, and its proof heavily relies on the preservation of context extension and the elimination rules associated with the type formers \cite{2ltt}. 

With the foundational knowledge established thus far, we are now equipped to delve into the discussion of potential models of 2LTT that satisfy the condition of having a cofibrant exo-nat. However, before proceeding to the subsequent section, where this discussion takes place, let us first provide the semantic definition of ``cofibrancy" for an exo-type.

\begin{defn}\label{cofib-semantics}
Let $(\C, \Ty^e, \Tm^e, \Ty^f, \Tm^f)$ be a model of 2LTT with conversion $\coer:\Ty^f\rightarrow \Ty^e$. We say an exo-type $A:\Ty^e(\Gamma)$ is \emph{cofibrant} if for any $\Delta : \C$ and $\sigma : \Delta \rightarrow \Gamma$

\begin{enumerate}
\item there is a map, natural in $\Delta$,
\[\Theta^{\Ty}_{\Delta}:\Ty^f(\Delta.A[\sigma])\rightarrow \Ty^f(\Delta) \] such that for any $Y:\Ty^f(\Delta.A[\sigma])$ we have the following isomorphism natural in $\Delta$: \[\coer_{\Delta}(\Theta^{\Ty}_{\Delta}(Y))\cong {\prod}^e(A,\coer_{\Delta.A[\sigma]}(Y)), \]
\item and there is a map, natural in $\Delta$,
\[\Theta^{\Tm}_{\Delta}:\Tm^f\left(\Delta.A[\sigma],\iscontr(Y)\right) \rightarrow \Tm^f\left(\Delta,\iscontr(\Theta_{\Delta}^{\Ty}(Y))\right).\] In other words, if $Y$ is contractible, then so is $\Theta^{\Ty}_{\Gamma}(Y)$.
\end{enumerate}
\end{defn}

\begin{remark}
This does not represent a \emph{direct} translation of the internal definition; rather, it can be seen as a \emph{universe-free} adaptation of it. In Definition \ref{cofib-defn}, the quantification is over specific types within a particular universe. Externally, we can express it in terms of \emph{all types}. Consequently, the external version holds slightly more strength.
\end{remark}

The naturality conditions give the following: In Figure \ref{naturality-cofib-type}, all vertical arrows are context substitutions. When the commutative sides are appropriately composed, the top and bottom isomorphisms are equal, which can be expressed as the cube ``commuting". In Figure \ref{naturality-cofib-term}, the vertical arrows are context substitutions, and the square is commutative.

\begin{figure}
\begin{center}
\begin{tikzcd}
                                                                                                  & {\Ty^f(\Delta.A[\sigma])} \arrow[rr, "\Theta^{\Ty}_{\Delta}", near start] \arrow[dd] \arrow[ld, "{\coer_{\Delta.A[\sigma]}}"']                &                          & \Ty^f(\Delta) \arrow[dd] \arrow[ld, "\coer_{\Delta}"'] \\
{\Ty^e(\Delta.A[\sigma])} \arrow[rr, "{\prod^e(A[\sigma],\_)}", near end, crossing over] \arrow[dd]            &                                                                                                                                             & \Ty^e(\Delta)  &                                                        \\
                                                                                                  & {\Ty^f(\Upsilon.A[\sigma\circ \tau])} \arrow[rr, "\Theta^{\Ty}_{\Upsilon}", near start] \arrow[ld, "{\coer_{\Upsilon.A[\sigma\circ \tau]}}"'] &                          & \Ty^f(\Upsilon) \arrow[ld, "\coer_{\Upsilon}"']        \\
{\Ty^e(\Upsilon.A[\sigma\circ \tau])} \arrow[rr, "{\prod^e(A[\sigma\circ \tau],\_)}", near end] &                                                                                                                                             & \Ty^e(\Upsilon)  \arrow[from=uu, crossing over]     &                                                       
\end{tikzcd}
\end{center}
\caption{Naturality condition for $\Theta^{\Ty}$, where $\sigma:\Delta\rightarrow \Gamma$ and $\tau : \Upsilon \rightarrow \Delta$ in $\C$.}\label{naturality-cofib-type}
\end{figure}

\begin{figure}
\begin{center}
\begin{tikzcd}
{\Tm^f\left(\Delta.A[\sigma],\iscontr(Y)\right)} \arrow[rr, "\Theta^{\Tm}_{\Delta}"] \arrow[dd]         &  & {\Tm^f\left(\Delta,\iscontr(\Theta_{\Delta}^{\Ty}(Y))\right)} \arrow[dd] \\
                                                                                                        &  &                                                                          \\
{\Tm^f \left(\Upsilon.A[\sigma\circ\tau],\iscontr(Y[\tau])\right)} \arrow[rr, "\Theta^{\Tm}_{\Upsilon}"'] &  & {\Tm^f\left(\Upsilon,\iscontr(\Theta_{\Upsilon}^{\Ty}(Y[\tau]))\right)}           
\end{tikzcd}
\end{center}
\caption{Naturality condition for $\Theta^{\Tm}$, where $\sigma:\Delta\rightarrow \Gamma$ and $\tau : \Upsilon \rightarrow \Delta$ in $\C$.}\label{naturality-cofib-term}
\end{figure}

\section{Models with cofibrant exo-nat}\label{modelsec}

In the following definition, recall that all products exist in the category of sets.

\begin{defn}\label{exo-nat-prod}
Let $(\C, \Ty, \Tm)$ be a CwF with enough type formers. We say $\C$ has \emph{exo-nat products} if there is a map $\Omega_{\Gamma}:\prod_{\mathbf{N}}\Ty(\Gamma)\rightarrow \Ty(\Gamma)$, where $\mathbf{N}$ is the external natural numbers, such that for any $Y : \prod_{\mathbf{N}}\Ty(\Gamma)$ we have
\begin{itemize}
\item[1)] the set $\Tm(\Gamma, \Omega_{\Gamma}(Y))$ is isomorphic to the categorical product of the sets $\Tm(\Gamma, Y_a)$ for each $a:\mathbf{N}$, namely, we have
\begin{center}
\begin{tikzcd}
{\Tm^f(\Gamma, \Omega_{\Gamma}(Y))} \arrow[rr, "\phi", bend left] & \cong & { \prod_{a:\mathbf{N}} \Tm^f(\Gamma, Y_a)} \arrow[ll, "\psi", bend left]
\end{tikzcd}
\end{center}
\item[2)] if $d,c : \prod_{a:\mathbf{N}} \Tm(\Gamma,Y_a)$ are such that there is a term in the type $\Id(d_a,c_a)$ as being terms of $Y_a:\Ty(\Gamma)$, then there is a term in the type $\Id(\psi(d),\psi(c))$ as being terms of $\Omega_{\Gamma}(Y)$,
\end{itemize}
and all these are natural in $\Gamma$.
\end{defn}

In simpler terms, the first requirement stated in Definition \ref{exo-nat-prod} ensures that $\prod^e_{a:\NN^e} Y(a)$ has a fibrant match, while the second requirement ensures that the $\funext$ for cofibrant exo-types holds.

\begin{example}\label{good-modcat}
Let $\C$ be a good model category \cite{lum-shul}. Define $\Ty(\Gamma)$ as the set of fibrations over $\Gamma$ (with suitable coherence conditions \cite{local-univ}). Define for $A: \Ty(\Gamma)$ the set $\Tm(\Gamma,P)$ as the hom-set $\sfrac{\C}{\Gamma}[\Gamma, \Gamma.A ]$.

Since $\Ty(\Gamma)$ is closed under countable products, we can take $\Omega_{\Gamma}(Y):=\prod_{a:\mathbf{N}}Y_a$, and there is a clear bijection between $\sfrac{\C}{\Gamma}[\Gamma, \Omega_{\Gamma}(Y)]$ and $\prod_{a:\mathbf{N}} (\sfrac{\C}{\Gamma}[\Gamma, Y_a])$, the first requirement in Definition \ref{exo-nat-prod} holds. 

Suppose $d,c : \prod_{a:\mathbf{N}} \Tm(\Gamma,Y_a)$ are such that there is a term in the type $\Id(d_a,c_a)$ as being terms of $Y_a:\Ty(\Gamma)$. In that model, it means $d_a$ and $c_a$, as being maps $\Gamma \rightarrow Y_a$, are right homotopic. That is, there are maps $p_a:\Gamma \rightarrow {Y_a}^I$ such that the following diagram commutes:
\begin{center}
\begin{tikzcd}
                                   & {{Y_a}^I} \arrow[d]  \\
\Gamma \arrow[ru, "p_a"] \arrow[r] & {{Y_a}\times {Y_a}}
\end{tikzcd}.
\end{center}
Since $Y_a$ is fibrant (as being in the slice category) and $\Gamma$ is cofibrant (as being an object of a good model category), by a standard lemma (Corollary 1.2.6 in \cite{hovey}), we have $d_a$ and $c_a$ are also left homotopic. Namely, the following diagram commutes:
\begin{center}
\begin{tikzcd}
\Gamma + \Gamma \arrow[r] \arrow[d] & {Y_a} \\
\Gamma' \arrow[ru]                  &       
\end{tikzcd}
\end{center}
where $\Gamma'$ is a cylinder object for $\Gamma$ fixed for all $a:\mathbf{N}$. This induces a left homotopy between $d$ and $c$, namely, we have:
\begin{center}
\begin{tikzcd}
\Gamma + \Gamma \arrow[r] \arrow[d] & {\prod_{a:\mathbf{N}} Y_a} \\
\Gamma' \arrow[ru]                  &       
\end{tikzcd}.
\end{center}
Now by the same lemma, we have $d$ and $c$ are right homotopic, that is, we have $p:\Gamma \rightarrow {(\prod_{a:\mathbf{N}} Y_a)^I} $ such that the following diagram commutes:
\begin{center}
\begin{tikzcd}
                                   & {(\prod_{a:\mathbf{N}} Y_a)^I} \arrow[d]  \\
\Gamma \arrow[ru, "p"] \arrow[r] & {\prod_{a:\mathbf{N}} Y_a \times \prod_{a:\mathbf{N}} Y_a}
\end{tikzcd}.
\end{center}
This means that there is a term in the type $\Id(d,c)$ as being terms of $\prod_{a:\mathbf{N}}Y_a$. So the second requirement in Definition \ref{exo-nat-prod} holds. We omit the details, but the naturality conditions follow from the coherence conditions on $\Ty$. \qed
\end{example}

In Theorem \ref{cofib-model}, we provide a class of two-level CwFs that satisfy the axiom that $\NN^e$ is a cofibrant exo-type. This is the main result of this section.

\begin{thm}\label{cofib-model}
If $(\C, \Ty, \Tm)$ is CwF (with sufficient type formers) has exo-nat products, the corresponding two-level CwF obtained by Definition \ref{presheaf-twolevel} satisfies the axiom that $\NN^e$ is a cofibrant exo-type.
\end{thm}

\begin{proof}
Recall that $\NN^e : \widehat{\Ty}(1_{\widehat{\C}})$. For any $P :\widehat{\C}$, the context morphism $\sigma : P \rightarrow 1_{\widehat{\C}}$ is unique, so we omit substitutions over such morphisms and write $P.\NN^e$ instead of $P.\NN^e[\sigma]$.

First, we define the map $\Theta^{\Ty}_{P}:\Ty^f(P.\NN^e)\rightarrow \Ty^f(P)$. For any $Y : \Ty^f(P.\NN^e) = \widehat{\C}(P.\NN^e, \Ty)$, we need $\Theta^{\Ty}_{P}(Y): \Ty^f(P)=\widehat{\C}(P, \Ty)$. We denote $\Theta^{\Ty}_{P}(Y)$ by $\tilde{Y}$ for easier reading. Now, for $\Gamma : \C$ and $x : P_{\Gamma}$, we define \[\tilde{Y}_{\Gamma}(x):= \Omega_{\Gamma}\left( (Y_{\Gamma}(x,n))_n\right)\] where the map $\Omega_{\Gamma} : \prod_{\mathbf{N}} \Ty(\Gamma) \rightarrow \Ty(\Gamma)$ is obtained by the assumption of having exo-nat products, Definition \ref{exo-nat-prod}. Since we have $Y_{\Gamma}(x,n) : \Ty(\Gamma)$, the definition makes sense. We want to show \[ \coer_P(\tilde{Y}) \cong {\prod}^e(\NN^e, \coer_{P.\NN^e}(Y)).\] Both are elements in $\widehat{\Ty}(P)$, namely, presheaves over $\int P$. Thus, it is enough to define a natural transformation $G:  \coer_P(\tilde{Y}) \rightarrow \prod^e(\NN^e, \coer_{P.\NN^e}(Y))$ such that for any $\Gamma : \C$ and $x : P_{\Gamma}$, the map \[G_{\Gamma,x} : \coer_P(\tilde{Y})(\Gamma,x) \rightarrow {\prod}^e(\NN^e, \coer_{P.\NN^e}(Y)) (\Gamma,x)\] is an isomorphism of sets, namely, a bijection. If we elaborate on these further, we obtain the following. By definition, $\coer_P(\tilde{Y})(\Gamma,x)=\Tm(\Gamma, \tilde{Y}_{\Gamma}(x))$. Also, ${\prod}^e(\NN^e, \coer_{P.\NN^e}(Y)) (\Gamma,x)$ consists of the elements 
\[f : \prod_{\substack{\Delta : \C \\ \sigma : \Delta \rightarrow \Gamma \\ n : \NN^e (\Delta, x[\sigma])}} \coer_{P.\NN^e}(Y)(\Delta, x[\sigma], n)\Big(=\Tm(\Delta, Y_{\Delta}(x[\sigma],n))\Big)\] such that if $\Upsilon : \C$ and $\tau : \Upsilon \rightarrow \Delta$, then $f(\Delta, \sigma, n)[\tau]=f(\Upsilon, \sigma \circ \tau, n[\tau])$. By the definition of $\NN^e$, we have $\NN^e(\Delta,x[\sigma])=\mathbf{N}$, external natural number set, and having exo-nat products provides us  \[\prod_{n :\mathbf{N}} \Tm(\Delta, Y_{\Delta}(x[\sigma],n)) \cong \Tm(\Delta, \tilde{Y}(x[\sigma])).\] Thus, the range of $G_{\Gamma,x}$ can be written as \[f : \prod_{\substack{\Delta : \C \\ \sigma : \Delta \rightarrow \Gamma}} \Tm(\Delta, \tilde{Y}_{\Delta}(x[\sigma]))\] such that if $\Upsilon : \C$ and $\tau : \Upsilon \rightarrow \Delta$, then $f(\Delta, \sigma)[\tau]=f(\Upsilon, \sigma \circ \tau)$. This elaboration allows us to easily perceive that this function is a bijection because it is a standard application of the Yoneda Lemma. The naturality condition of this operation is also easily satisfied because the substitution is functorial. Thus, we have confirmed the first stage of our claim. It remains to handle the contractibility part.

Basically, we need to show that for the center of contraction $c : \Tm^f(P.\NN^e, Y)$ and the identity terms in $\Id(c,d)$ for any other terms $d : \Tm^f(P.\NN^e, Y)$, we can find (naturally) a center of contraction $\tilde{c}:\Tm^f(P,\tilde{Y})$ and an identity term $\Id(\tilde{c},\tilde{d})$ for any other terms $\tilde{d}:\Tm^f(P,\tilde{Y})$. With a similar elaboration on terms, we have $\Tm^f(P.\NN^e, Y)= \widehat{\Tm}(P.\NN^e, \coer_{P.\NN^e}(Y))$ and $\Tm^f(P,\tilde{Y})=\widehat{\Tm}(P, \coer_P(\tilde{Y}))$. Therefore, we know \[c : \prod_{\substack{\Gamma : \C \\ x : P_{\Gamma} \\ n : \NN^e (\Gamma, x)}} \coer_{P.\NN^e}(Y)(\Gamma, x, n)\Big(=\Tm(\Gamma, Y_{\Gamma}(x,n))\Big)\] is a center of contraction, for any such term $d$, we have a term in $\Id(c,d)$, and we need \[\tilde{c} : \prod_{\substack{\Gamma : \C \\ x : P_{\Gamma}}} \coer_{P}(\tilde{Y})(\Gamma, x)\left(=\Tm(\Gamma, \tilde{Y}_{\Gamma}(x))\right)\] as a center of contraction, and related contracting terms. However, this is exactly the second criterion in Definition \ref{exo-nat-prod}, and we have already assumed it. 

Therefore, $\NN^e$ is a cofibrant exo-type in the presheaf two-level CwF.
\end{proof}

\begin{remark}
Example \ref{good-modcat} also enables us to construct a two-level CwF out of the class in the theorem that satisfies the axiom. Indeed, we can take $\Ty^e(\Gamma)$ as the set of all morphisms over $\Gamma$, and $\Tm^e$ as the same as $\Tm$, and obtain a two-level CwF $(\C, \Ty^e, \Tm^e, \Ty, \Tm)$ with the conversion $\coer:\Ty\rightarrow \Ty^e$ as being inclusion. It is then enough to take the map $\Theta_{\Gamma}(Y)$ in Definition \ref{cofib-semantics} as equal to $\prod_{a:\mathbf{N}}Y_a$. 
\end{remark}

\section{Future directions}

As previously mentioned, this formalisation project aims to move a study about 2LTT \cite{UPpaper} to Agda. In addition to definitions and results here, we also formalised \emph{exo-categories} and \emph{diagram signatures} in that study. More will be added in the future. 

We also plan to generalize the results about cofibrancy and sharpness. Natural numbers, lists, and binary-trees are all inductive types. The general class of such inductive types is called \emph{W-types}. Similar to the cofibrant exo-nat axiom, we have been studying on possible conditions (or axioms) related to W-types to obtain criteria for cofibrant and sharp W-types. Currently, a study on W-types in 2LTT has not been conducted yet; thus, we plan to work on this open problem. 

We have been studying to improve the Agda library. The experimental feature of Agda we used reveals also some bugs; hence, we plan to solve these issues to obtain precise consistency. Furthermore, it is not unreasonable to think that this study will offer new ideas about the concepts specific to 2LTT.

\end{document}